\documentclass[a4paper,UKenglish,cleveref,autoref,colorlinks, allcolors=blue]{lipics-v2021}
\usepackage[utf8]{inputenc}
\usepackage{amsmath, amssymb, amsthm}
\usepackage{color}
\usepackage{todonotes}
\nolinenumbers
\usepackage{url}
\usepackage{hyperref}
\usepackage{enumitem}
\usepackage{graphicx}
\usepackage{tikz}
\usepackage{mathtools}
\usepackage{thm-restate}
\usetikzlibrary{calc,arrows.meta}
\usepackage{dsfont}

\theoremstyle{definition}

\numberwithin{theorem}{section}
\numberwithin{lemma}{section}
\numberwithin{claim}{section}
\numberwithin{corollary}{section}
\hideLIPIcs

\DeclareMathOperator{\tb}{b}
\DeclareMathOperator{\vc}{vc}
\DeclareMathOperator{\tw}{tw}
\DeclareMathOperator{\pw}{pw}
\DeclareMathOperator{\wtpw}{wtpw}

\newcommand{\defparaproblem}[4]{
 \vspace{2mm}
\noindent\fbox{
 \begin{minipage}{0.96\textwidth}
 \begin{tabular*}{\textwidth}{@{\extracolsep{\fill}}lr} #1 & \\ \end{tabular*}
 {\textbf{Input:}} #2 \\
 {\textbf{Parameter:}} #3 \\
 {\textbf{Question:}} #4
 \end{minipage}
 }
 \vspace{2mm}
}

\newcommand{\defproblem}[3]{
 \vspace{2mm}
\noindent\fbox{
 \begin{minipage}{0.96\textwidth}
 \begin{tabular*}{\textwidth}{@{\extracolsep{\fill}}lr} #1 & \\ \end{tabular*}
 {\textbf{Input:}} #2 \\
 {\textbf{Question:}} #3
 \end{minipage}
 }
 \vspace{2mm}
}

\newcommand{\IMCF}{\textsc{Integer Multicommodity Flow}}
\newcommand{\ILCF}[1]{\textsc{Integer $#1$-Commodity Flow}}
\newcommand{\UILCF}[1]{\textsc{Undirected Integer $#1$-Commodity Flow}}
\newcommand{\UILCFME}[1]{\textsc{Undirected Integer $#1$-Commodity Flow with Monochrome Edges}}
\newcommand{\ILCFME}[1]{\textsc{Integer $#1$-Commodity Flow with Monochrome Edges}}
\newcommand{\chgad}[1]{$#1$-Gate}
\newcommand{\chgadn}[0]{Gate}
\newcommand{\tcmc}[0]{\textsc{Tree Chained Multicolour Clique}}
\title{The Parameterised Complexity of  Integer Multicommodity Flow%
\thanks{The research of Isja Mannens was supported by the project CRACKNP that has received funding from the European Research Council (ERC) under the European Union’s Horizon 2020 research and innovation programme (grant agreement No 853234). The research of Jelle Oostveen was supported by the NWO grant OCENW.KLEIN.114 (PACAN).}}
\titlerunning{The Parameterised Complexity of Integer Multicommodity Flow}

\author{Hans L. Bodlaender}{Utrecht University, The Netherlands}{h.l.bodlaender@uu.nl}{https://orcid.org/0000-0002-9297-3330}{}
\author{Isja Mannens}{Utrecht University, The Netherlands}{i.m.e.mannens@uu.nl}{https://orcid.org/0000-0003-2295-0827}{}
\author{Jelle J. Oostveen}{Utrecht University, The Netherlands}{j.j.oostveen@uu.nl}{https://orcid.org/0009-0009-4419-3143}{}
\author{Sukanya Pandey}{Utrecht University, The Netherlands}{s.pandey1@uu.nl}{https://orcid.org/0000-0001-5728-1120}{}
\author{Erik Jan van Leeuwen}{Utrecht University, The Netherlands}{e.j.vanleeuwen@uu.nl}{https://orcid.org/0000-0001-5240-7257}{}
\authorrunning{Bodlaender et al.}


\date{\today}

\begin{document}
\maketitle
\begin{abstract}
The \textsc{Integer Multicommodity Flow} problem has been studied extensively in the literature. However, from a parameterised perspective, mostly special cases, such as the \textsc{Disjoint Paths} problem, have been considered. Therefore, we investigate the parameterised complexity of the general \textsc{Integer Multicommodity Flow} problem. We show that the decision version of this problem on directed graphs for a constant number of commodities, when the capacities are given in unary, is XNLP-complete with pathwidth as parameter and XALP-complete with treewidth as parameter. When the capacities are given in binary, the problem is NP-complete even for graphs of pathwidth at most~$13$. We give related results for undirected graphs. These results imply that the problem is unlikely to be fixed-parameter tractable by these parameters.

In contrast, we show that the problem does become fixed-parameter tractable when weighted tree partition width (a variant of tree partition width for edge weighted graphs) is used as parameter.
\end{abstract}

\section{Introduction}
The \textsc{Multicommodity Flow} problem is the generalisation of the textbook flow problem where instead of just one commodity, multiple different commodities have to be transported through a network. The problem models important operations research questions (see e.g.~\cite{Wang1}). 
Although several optimisation variants of this problem exist~\cite{Wang1}, we consider only the variant where for each commodity, a given amount of flow (the demand) has to be sent from the commodity's source to its sink, subject to a capacity constraint on the total amount flow through each arc. 
The nature and computational complexity of the problem is highly influenced by the graph (undirected or directed, its underlying structure) and the capacities, demands, and flow value (integral or not, represented in unary or binary). When the flow values are allowed to be fractional, the problem can be trivially solved through a linear program 
(see e.g.~\cite{Karakostas08,KorteV2000}).

We focus on \IMCF{}, where all the given capacities and demands are integers, and the output flow must also be integral. The \IMCF{} problem is widely studied and well known to be NP-hard even if all capacities are~$1$, on both directed and undirected graphs, even when there are only two commodities~\cite{NPCundirected}. On directed graphs, it is NP-hard even for two commodities of demand~$1$~\cite{FortuneHW80}. These strong hardness results have led to a wide range of heuristic solution methods being investigated as well as a substantial body of work on approximation algorithms. For surveys, see e.g., \cite{Barnhart2009,Wang1,Wang2}.

An important special case of \IMCF{} and the main source of its computational hardness is the \textsc{Edge Disjoint Paths} problem. It can be readily seen that \IMCF{} is equivalent to \textsc{Edge Disjoint Paths} when all capacities and demands are~$1$. Indeed, all aforementioned hardness results stem from this connection. The \textsc{Edge Disjoint Paths} problem has been studied broadly in its own right (see e.g.\ the surveys by Frank~\cite{Frank1990} and Vygen~\cite{Vygen1998}), including a large literature on approximation algorithms. See, amongst others~\cite{GuruswamiKRSY03,ShepherdV17} for further hardness and inapproximability results. On undirected graphs, \textsc{Edge Disjoint Paths} is fixed-parameter tractable parameterised by the number of source-sink pairs~\cite{RobertsonS95b,KawarabayashiKR12}.

Investigation of the parameterised complexity of \textsc{Edge Disjoint Paths} has recently been continued by considering structural parameterisations. Unfortunately, the problem is NP-hard for graphs of treewidth~$2$~\cite{NishizekiVZ01} and even for graphs with a vertex cover of size~$3$~\cite{FleszarMS18}. It is also W[1]-hard parameterised by the size of a vertex set whose removal leaves an independent set of vertices of degree~$2$~\cite{GanianO21}. From an algorithmic perspective, Ganian and Ordyniak~\cite{GanianO21} showed that \textsc{Edge Disjoint Paths} is in XP parameterised by tree-cut width. Zhou et al.~\cite{ZhouTN00} give two XP algorithms for \textsc{Edge Disjoint Paths} for graphs of bounded treewidth: one for when the number of paths is small, and one for when a specific condition holds on the pairs of terminals. 
Ganian et al.~\cite{GanianOR21} give an FPT algorithm parameterised by the treewidth and degree of the graph. Friedrich et al.~\cite{FriedrichIKMZ22a,FriedrichIKMZ22} give approximation algorithms for multicommodity flow on graphs of bounded treewidth.

These results naturally motivate the question: 
\begin{quote}
\textit{What can we say about the parameterised complexity of the general \IMCF{} problem under structural parameterisations?}
\end{quote}
We are unaware of any explicit studies in this direction. We do note that the result of Zhou et al.~\cite{ZhouTN00} implies an XP algorithm on graphs of bounded treewidth for a bounded number of commodities if the capacities are given in unary. We are particularly interested in whether this result can be improved to an FPT algorithm, which is hitherto unknown.

\subsection{Our Setting and Contributions}
We consider the \IMCF{} problem for a small, fixed number of commodities. In particular, \ILCF{\ell} is the variant in which there are $\ell$ commodities. Furthermore, we study the setting where some well-known structural parameter of the input graph, particularly its pathwidth or treewidth, is small. 

\begin{table}[tb]
    \centering
    \begin{tabular}{|c|c|c|}
    \hline
        Parameter & unary capacities & binary capacities  \\ \hline
        pathwidth & XNLP-complete & para-NP-complete \\
        treewidth & XALP-complete & para-NP-complete \\
        weighted tree partition width & FPT (1) & FPT (1) \\
        vertex cover & (2); in XP & (2); open \\\hline
    \end{tabular}
    \caption{Overview of our results for \ILCF{2}. para-NP-complete = NP-complete for fixed value of parameter. (1) = capacities of arcs inside bags can be arbitrary, capacities of arcs between bags are bounded by weighted tree partition width. (2) Approximation, see Theorem~\ref{theorem:vcapprox}; conjectured in FPT. For the undirected case, the same results hold, except that for the para-NP-completeness for the parameters pathwidth and treewidth, we need a third commodity.}
    \label{table:results}
\end{table}

Our main contribution is to show that \ILCF{2} is unlikely to be fixed-parameter tractable parameterised by treewidth and or by pathwidth.
Instead of being satisfied with just a W[$t$]-hardness result for some $t$ or any $t$, we seek stronger results using the recently defined complexity classes XNLP and XALP. 
An overview of our results can be found in Table~\ref{table:results}. 

XNLP is the class of parameterised problems that can be solved on a non-deterministic Turing machine in $f(k)|x|^{O(1)}$ time and $f(k)\log |x|$ memory for a computable function $f$, where $|x|$ is the size of the input~$x$. 
The class XNLP (under a different name) was first introduced by Elberfeld et al.~\cite{ElberfeldST15}. 
Bodlaender et al.~\cite{BodlaenderCW22a,BodlaenderGJJL22,BodlaenderGNS22a} showed a number of problems to be XNLP-complete with pathwidth as parameter. In particular, \cite{BodlaenderCW22a} gives XNLP-completeness proofs for several flow problems with pathwidth as parameter.

In this work, we prove XNLP-completeness (and stronger) results for \ILCF{\ell}. These give a broad new insight into the complexity landscape of \IMCF{}. We distinguish how the capacities of arcs and edges are specified: these can be given in either unary or binary. First, we consider the unary case:

\begin{restatable}{theorem}{twocomXNLP}\label{thm:2comXNLP}
\ILCF{2} with capacities given in unary, parameterised by pathwidth, is XNLP-complete.
\end{restatable}

\begin{restatable}{theorem}{UntwocomXNLP}\label{thm:UILCF2}
\UILCF{2} with capacities given in unary, parameterised by pathwidth, is XNLP-complete.
\end{restatable}

These hardness results follow by reduction from the XNLP-complete {\sc Chained Multicoloured Clique} problem~\cite{BodlaenderGJPP22a}, a variant of the perhaps more familiar {\sc Multicoloured Clique} problem~\cite{FellowsHRV09}. We follow a common strategy in such reductions, using vertex selection and edge verification gadgets. However, a major hurdle is to use flows to select vertices and verify the existence of edges to form the sought-after cliques. To pass this hurdle, we construct gadgets that use Sidon sets as flow values, combined with gadgets to check that a flow value indeed belongs to such a Sidon set.

For the parameter treewidth, we are able to show a slightly stronger result. 
Recently, Bodlaender et al.~\cite{BodlaenderGJPP22a} introduced the complexity class XALP, which is the class of parameterised problems that can be solved on a non-deterministic Turing machine
that has access to an additional stack,
in $f(k)|x|^{O(1)}$ time and $f(k)\log |x|$ space (excluding the space used by the stack), for a computable function $f$, where $|x|$ again denotes the size of the input~$x$.
Many problems that are XNLP-complete with pathwidth as parameter are XALP-complete with treewidth as parameter. We show that this phenomenon also holds for the studied \IMCF{} problem:

\begin{restatable}{theorem}{twocomXALP}\label{thm:2comXALP}
\ILCF{2} with capacities given in unary, parameterised by treewidth, is XALP-complete.
\end{restatable}

The reduction is from the XALP-complete \textsc{Tree-Chained Multicoloured Clique} problem~\cite{BodlaenderGNS22a} and follows similar ideas as the above reduction.  Combining techniques of the proofs of Theorems~\ref{thm:UILCF2} and \ref{thm:2comXALP} gives the following result.

\begin{restatable}{theorem}{UntwocomXALP}\label{thm:U2comXALP}
\UILCF{2} with capacities given in unary, parameterised by treewidth, is XALP-complete.
\end{restatable}
Assuming the \emph{Slice-wise Polynomial Space Conjecture}~\cite{PilipczukW18, BodlaenderGJJL22}, these results show that XP-algorithms for \ILCF{2} or \UILCF{2} for graphs of small pathwidth or treewidth cannot use only $f(k)|x|^{O(1)}$ memory. Moreover, the XNLP- and XALP-hardness implies these problems are $W[t]$-hard for all positive integers $t$.

If the capacities are given in binary, then the problems become
even harder.

\begin{restatable}{theorem}{twocomNPC}\label{theorem:binarypathwidth}
\ILCF{2} with capacities given in binary is NP-complete for graphs of pathwidth at most~$13$.
\end{restatable}

\begin{restatable}{theorem}{UnthreecomNPC}\label{theorem:undirected3}
\UILCF{3} with capacities given in binary is NP-complete for graphs of pathwidth at most~$18$.
\end{restatable}

Finally, we consider a variant of the {\IMCF} problem where the flow must be \emph{monochrome}, i.e.\ a flow is only valid when no edge carries more than one type of commodity. Then, we obtain hardness even for parameterisation by the vertex cover number of the graph, for both variants of the problem.

\begin{restatable}{theorem}{DiILCFME}
    \ILCFME{2} is NP-hard for binary weights and vertex cover number~$6$, and W[1]-hard for unary weights when parameterised by the vertex cover number.
\end{restatable}

\begin{restatable}{theorem}{UnILCFME}\label{thm:MonoEdge}
\UILCFME{2} is NP-hard for binary weights and vertex cover number~$6$, and W[1]-hard for unary weights when parameterised by the vertex cover number.
\end{restatable}

To complement our hardness results, we prove two algorithmic results. Bodlaender et al.~\cite{BodlaenderCW22a} had given FPT algorithms for several flow problems,
using the recently defined notion of weighted tree partition width as parameter (see~\cite{BodlaenderCW22,BodlaenderCW22a}). Weighted tree partition width can be seen as a variant of the notion of \emph{tree partition width} for edge-weighted graphs, introduced by Seese~\cite{Seese85} in 1985 under the name \emph{strong treewidth}. See Section~\ref{sec:prelims} for formal definitions of these parameters. 
We note that the known hardness for the vertex cover number~\cite{FleszarMS18} implies that {\sc Edge Disjoint Paths} is NP-hard even for graphs of tree partition width $3$. Here, we prove the following:

\begin{restatable}{theorem}{UellCF}\label{thm:dwtpw}
The \ILCF{\ell} problem can be solved in  time $2^{2^{\tb^{3\ell b}}} n^{O(1)}$, where $\tb$ is the breadth of a given tree partition of the input graph.
\end{restatable}
\begin{restatable}{theorem}{UnILCFtb}\label{thm:uwtpw}
    The \UILCF{\ell} problem can be solved in  time $2^{2^{\tb^{3\ell b}}} n^{O(1)}$, where $\tb$ is the breadth of a given tree partition of the input graph.
\end{restatable}

For the standard \ILCF{2} problem with the vertex cover number of the input graph as parameter,
we conjecture that this problem is in FPT. As a partial result, we can give the following approximation algorithms. Let $\vc(G)$ denote the vertex cover number of a graph $G$.

\begin{restatable}{theorem}{vcapprox}\label{theorem:vcapprox}
There is a polynomial-time algorithm that, given an instance of \ILCF{2} on a graph $G$ with demands $d_1,d_2$, either outputs that there is no flow that meets the demands or outputs a $2$-commodity flow of value at least $d_i - O(\vc(G)^3)$ for each commodity $i\in [2]$.
\end{restatable}

\begin{restatable}{theorem}{Unvcapprox}
There is a polynomial-time algorithm that, given an instance of \UILCF{2} on a graph $G$ with demands $d_1,d_2$, either outputs that there is no flow that meets the demands or outputs a $2$-commodity flow of value at least $d_i - O(\vc(G)^3)$ for commodity $i\in [2]$.
\end{restatable}

%--------------------------------------------------
% PRELIMINARIES
%--------------------------------------------------

\section{Preliminaries}\label{sec:prelims}
In this paper, we consider both directed and undirected graphs. Graphs are directed unless explicitly stated otherwise. Arcs and edges are denoted as $vw$ (an arc from $v$ to $w$, or an edge with $v$ and $w$ as endpoints).

\subsection{Integers}
We use the interval notation for intervals of integers, e.g., $[-1,3] = \{-1,0,1,2,3\}$. We will simplify this notation for intervals that start at $1$, i.e. $[k] = [1,k]$. Moreover, we use $\mathbb{N} = \{1,2,\ldots\}$ and $\mathbb{N}_0 = \{0,1,2,\ldots\}$.

A \emph{Sidon set} is a set of positive integers $\{a_1, a_2, \ldots, a_n\}$ such that all pairs
have a different sum, i.e., when $a_i+a_{i'} = a_j + a_{j'}$ then $\{i,i'\}=\{j,j'\}$.  Sidon sets are
also Golomb rulers and vice versa --- in a Golomb ruler, pairs of different elements have unequal differences:
if $i\neq i'$ and $j\neq j'$, then $|a_i -a_{i'}|=|a_j-a_{j'}|$, then $\{i,i'\}=\{j,j'\}$. A construction by Erd\"{o}s and Tur\'{a}n~\cite{ErdosT41} for Sidon sets implies the following, cf.~the discussion in~\cite{BodlaenderW20}.

\begin{theorem} \label{thm:sidon}
A Sidon set of $n$ elements in $[4n^2]$ can be found in $O(n \sqrt{n})$ time and logarithmic space.
\end{theorem}

\subsection{Multicommodity Flow Problems}
We now formally define our flow problems. A \emph{flow network} is a pair $(G,c)$ of a directed (undirected) graph $G=(V,E)$ and a function $c: E \rightarrow \mathbb{N}_0$ that assigns to each arc (edge) a non-negative integer \emph{capacity}. We generally use $n = |V|$ and $m = |E|$.

For a positive integer $\ell$, an \emph{$\ell$-commodity flow} in a flow network with sources $s_1, \ldots, s_\ell \in V$ and sinks $t_1, \ldots, t_\ell \in V$ is a $\ell$-tuple of functions $f^1, \ldots, f^\ell: E \rightarrow \mathbb{R}_{\geq 0}$, that fulfils the following conditions:
\begin{itemize}
    \item \textbf{Flow conservation.} For all $i\in [\ell]$, $v\not\in\{s_i,t_i\}$, $\sum_{wv\in E} f^i(wv) = \sum_{vw\in E} f^i(vw) $.
    \item \textbf{Capacity.} For all $vw\in E, \sum_{i\in [\ell]} f^i(vw) \leq c(vw)$.
\end{itemize}
An $\ell$-commodity flow is an \emph{integer $\ell$-commodity flow} if for all $i\in [c]$, $vw\in E$, $f^i(vw) \in \mathbb{N}_0$. The \emph{value for commodity $i$} of an $\ell$-commodity flow equals
$\sum_{s_iw\in E} f^i(s_iw) - \sum_{ws_i\in E} f^i(ws_i)$. We shorten this to `flow' when it is clear from context what the value of $\ell$ is and whether we are referring to an integer or non-integer flow.

The main problem considered in the paper now is as follows:

\defproblem
    {\ILCF{\ell}}
    {A flow network $G=(V,E)$ with capacities $c$, sources $s_1, \ldots, s_\ell \in V$, sinks $t_1, \ldots, t_\ell \in V$, and demands $d_1, \ldots, d_\ell \in \mathbb{N}$.}
    {Does there exist an integer $\ell$-commodity flow in $G$ which has value $d_i$ for each commodity $i\in [\ell]$?}

The \IMCF{} problem is the union of all \ILCF{\ell} problems for all non-negative integers $\ell$.

For undirected graphs, flow still has direction, but the capacity constraint changes to:
\begin{itemize}
    \item \textbf{Capacity.} For all $vw\in E, \sum_{i\in [\ell]} f^i(vw) + f^i(wv) \leq c(vw)$.
\end{itemize}
The undirected version of the \ILCF{\ell} problem then is as follows:

\defproblem
    {\UILCF{\ell}}
    {An \emph{undirected} flow network $G=(V,E)$ with capacities $c$, sources $s_1, \ldots, s_\ell \in V$, sinks $t_1, \ldots, t_\ell \in V$, and demands $d_1, \ldots, d_\ell \in \mathbb{N}$.}
    {Does there exist an integer $\ell$-commodity flow in $G$ which has value $d_i$ for each commodity $i\in [\ell]$?}

Finally, we say that an $\ell$-commodity flow $f^1, \ldots, f^\ell$ is \emph{monochrome} if no arc (edge) has positive flow of more than one commodity. That is, if $f^i(e) > 0$ for some arc (or edge) $e$, then $f^{i'}(e) = 0$ for all $i' \in [\ell] \setminus \{i\}$. We can then immediately define monochrome versions of {\ILCF{\ell}} and {\UILCF{\ell}} in the expected way.

\subsection{Parameters}
We now define the various parameters and graph decompositions that we use in our paper.
A \emph{tree decomposition} of a graph $G=(V,E)$ is a pair 
$(\{X_i~|~i\in I\}, T)$, with $\{X_i~|~i\in I\}$ a family
of subsets (called \emph{bags}) of $V$, and $T=(I,F)$ a tree,
such that: $\bigcup_{i\in I} X_i=V$; for all $vw\in E$, there is
a bag $X_i$ with $v,w\in X_i$; and for all $v\in V$,
the nodes $i$ with $v\in X_i$ form a (connected) subtree of $T$. 
The \emph{width} of a tree decomposition  $(\{X_i~|~i\in I\}, T)$ equals $\max_{i\in I} |X_i|-1$, and
the \emph{treewidth} $\tw(G)$ of $G$ is the
minimum width of a tree decomposition of $G$. 

A tree decomposition
$(\{X_i~|~i\in I\}, T)$ is a \emph{path decomposition}, 
if $T$ is a path, and the \emph{pathwidth} $\pw(G)$ of a graph $G$ is the
minimum width of a path decomposition of $G$.

A \emph{tree partition} of a graph $G=(V,E)$ is a pair $(\{B_i~|~i\in I\}, T)$, with $\{B_i~|~i\in I\}$ a family
of subsets (called \emph{bags}) of $V$, and $T=(I,F)$ a tree, such that
\begin{enumerate}
    \item For each vertex $v\in V$, there is exactly one $i\in I$ with $v\in B_i$. (I.e., $\{B_i~|~i\in I\}$
    forms a partition of $V$, except that we allow that some bags are empty.)
    \item For each edge $vw\in E$, if $v\in B_i$ and $w\in B_{i'}$ then $i=i'$ or $ii'\in F$.
\end{enumerate}
The \emph{width} of a tree partition $(\{B_i~|~i\in I\}, T)$ equals $\max_{i\in I} |B_i|$, and the 
\emph{tree partition width} of a graph $G$ 
is the minimum width of a tree partition of $G$. 

The notion of weighted tree partition width is defined for edge-weighted graphs and originates in the work of Bodleander et al.~\cite{BodlaenderCW22a} (see also~\cite{BodlaenderCW22}). Let $G=(V,E)$ be a graph and suppose $c: E \rightarrow \mathbb{N}$ is an edge-weight function. The \emph{breadth}\footnote{This notion of breadth should not be confused with the breadth of a tree decomposition and the related notion of treebreadth~\cite{DraganK14}. The breadth of a tree decomposition is defined as the maximum radius of any bag of a tree decomposition.} of a tree partition $(\{B_i~|~i\in I\}, T)$ of $G$ equals the maximum of $\max_{i\in I} |B_i|$ and $\max_{ii'\in F} c(B_i,B_{i'})$ with $c(B_i,B_{i'}) = \sum_{e=vw, v\in B_i, w\in B_{i'}} c(e)$, i.e., the maximum sum of edge weights of edges between the bags of $T$. Then the \emph{weighted tree partition width} $\wtpw(G)$ of $G$ is the minimum breadth of any tree partition of $G$. 
For our application, we interpret the capacity function as the weight function for weighted tree partition width.

Finally, a \emph{vertex cover} of a graph $G$ is a set $X \subseteq V(G)$ such that $X \cap \{u,v\} \not= \emptyset$ for every edge $uv \in V(G)$. Then the \emph{vertex cover number} $\vc(G)$ of $G$ is the size of the smallest vertex cover of $G$.

We also use these parameters for directed graphs. In that case, the direction of edges is ignored, i.e.,
the treewidth, pathwidth, tree partition width, or vertex cover number of a directed graph equals that parameter for the underlying undirected graph.

\subsection{XNLP and XALP}
The class XNLP is the class of parameterised problems that can be solved on a non-deterministic Turing machine in $f(k)|x|^{O(1)}$ time and $f(k)\log |x|$ memory for a computable function $f$, where $|x|$ is the size of the input~$x$.
The class XALP is the class of parameterised problems that can be solved on a non-deterministic Turing machine
that has access to an additional stack,
in $f(k)|x|^{O(1)}$ time and $f(k)\log |x|$ space (excluding the space used by the stack) for a computable function $f$, where $|x|$ is the size of the input~$x$.

A {\em parameterized logspace reduction} or {\em pl-reduction}
from a parameterized problem $A\subseteq \Sigma^\ast \times \mathbb{N}$
to a parameterized problem $B\subseteq \Sigma^\ast \times \mathbb{N}$
is a function $f: \Sigma^\ast \times \mathbb{N} \rightarrow \Sigma^\ast \times \mathbb{N}$, such that
\begin{itemize}
    \item there is an algorithm that computes $f((x,k))$ in space $O(g(k) + \log n)$, with $g$ a computable
    function and $n=|x|$ the number of bits to denote $x$;
    \item  for all $(x,k) \in \Sigma_1^{\ast} \times \mathbb{N}$, $(x,k)\in A$ if and only if $f((x,k)) \in B$.
    \item there is a computable function $g$, such that for all $(x,k) \in \Sigma_1^{\ast} \times \mathbb{N}$, if $f(x,k)=(x',k')$, then $k'\leq g(k)$.
\end{itemize}

XNLP-hardness and XALP-hardness are defined with respect to pl-reductions. The main difference with the more standard parameterized reductions is that 
the computation of the reduction must be done with logarithmic space. In most
cases, existing parameterized reductions are also pl-reductions; logarithmic space is achieved by not storing intermediate results but recomputing these
when needed.

Hardness for XNLP (and XALP) has a number of consequences. One is the following 
conjecture due to Pilipczuk and Wrochna~\cite{PilipczukW18} that states that XP-algorithms for XNLP-hard problems are likely to 
use much memory.

\begin{conjecture}[Slice-wise Polynomial Space Conjecture \cite{PilipczukW18}] \label{con:slice}
    If parameterized problem $A$ is XNLP-hard, then there is no algorithm that solves
    $A$ in $|x'|^{f(k)}$ time and $f(k)|x|^{O(1)}$ space, for instances $(x,k)$, with
    $f$ a computable function, and $|x|$ the size of instance $x$.
\end{conjecture}

One can easily observe from the definitions that for all $t\in \mathbb{N}$,
$W[t] \subseteq$ XNLP. Thus, we have the following lemma.

\begin{lemma} \label{lem:wt}
If problem $A$ is XNLP-hard, that $A$ is hard for all classes $W[t]$, $t\in \mathbb{N}$.
\end{lemma}

Our hardness proofs start from two variations of the well-known \textsc{Multicoloured Clique} problem (see~\cite{FellowsHRV09}).

\defparaproblem
    {\textsc{Chained Multicoloured Clique}}
    {A graph $G=(V,E)$, a partition of $V$ into $V_1, \ldots, V_r$, such that $|i-j| \leq 1$ for each edge $uv \in E(G)$ with $u \in V_i$ and $v \in V_j$, and a function $c \colon V \to [k]$.}
    {$k$.}
    {Is there a set of vertices $W \subseteq V$ such that for all $i \in [r-1]$, $W \cap (V_i \cup V_{i+1})$ is a clique, and for each $i \in [r]$ and $j \in [k]$, there is a vertex $v \in W \cap V_i$ with $c(v) = j$?}

\defparaproblem
    {\textsc{Tree-Chained Multicoloured Clique}}
    {A graph $G=(V,E)$, a tree partition $(\{V_i\mid i\in I\},T=(I,F))$ with $T$ a tree of maximum degree~$3$, and a function $c \colon V \to [k]$.}
    {$k$.}
    {Is there a set of vertices $W \subseteq V$ such that for all $ii' \in F$, $W \cap (V_i \cup V_{i'})$ is a clique, and for each $i \in I$ and $j \in [k]$, there is a vertex $v \in W \cap V_i$ with $c(v) = j$?}

\begin{theorem}[From~\cite{BodlaenderGJPP22a} and \cite{BodlaenderGNS22a}] \label{thm:cmc-tcmc}
    \textsc{Chained Multicoloured Clique} is XNLP-complete, and \textsc{Tree-Chained Multicoloured Clique}
    is XALP-complete.
\end{theorem}

\section{Hardness results}
In this section, we give the proofs of our hardness results. The section is partitioned into three parts. We start by giving the results for the case of unary capacities, parameterised by pathwidth and parameterised by treewidth, both for the directed and undirected cases. This is followed by our results for binary capacities in these settings. Finally, we give the results for graphs of bounded vertex cover, for both the unary and binary case.

%---------------------------------------------
% UNARY
%---------------------------------------------

\subsection{Unary Capacities}
We prove our hardness results for {\IMCF} with unary capacities, parameterised by pathwidth and parameterised by treewidth. We aim to reduce from \textsc{Chained Multicoloured Clique} (for the parameter pathwidth) and \textsc{Tree-Chained Multicoloured Clique} (for the parameter treewidth). We first introduce a number of gadgets: subgraphs that fulfil certain properties and that are used in the hardness constructions. After that, we give the hardness results for directed graphs, followed by reductions from the directed case to the undirected case.

Before we start describing the gadgets, it is good to know that all constructions will have disjoint sources and sinks for the different commodities. We will set the demands for each commodity equal to the total capacity of the outgoing arcs from the sources, which is equal to the total capacity of the incoming arcs to the sinks. Thus, the flow over such arcs will be equal to their capacity.

Furthermore, throughout this section, our constructions will have two commodities. We name the commodities 1 and 2, with sources $s_1,s_2$ and sinks $t_1,t_2$, respectively.

\subsubsection{Gadgets}
We define two different types of (directed) gadgets. Given an integer $a$, the \emph{\chgad{a} gadget} either can move $1$ unit of flow from one commodity from left to right, or at most $a$ units of flow from the other commodity from top to bottom, but not both. Hence, it models a form of choice. This gadget will grow in size with $a$, and thus will only be useful if the input values are given in unary. Given a set $S$ of integers and a large integer $L$ (larger than any number in $S$), the \emph{$(S,L)$-Verifier} is used to check if the flow over an arc belongs to a number in $S$. The \chgad{a} gadget is used as a sub-gadget in this construction. In our reduction, later, we will use appropriately constructed sets $S$ to select vertices or to check for the existence of edges. Both types of gadget have constant pathwidth, and thus constant treewidth.

When describing the gadgets and proving that their tree- or pathwidth is bounded, it is often convenient to think of them as puzzle pieces being placed in a bigger mold. Formally, a \emph{(puzzle) piece} is a directed (multi-)graph $H$ given with a set $B^- \subseteq V(H)$ of vertices that have in total $a$ incoming arcs without tail (\emph{entry arcs}) and a set $B^+ \subseteq V(H)$ of vertices that have in total $b$ outgoing arcs without head (\emph{exit arcs}). The sets $B^-$ and $B^+$ are disjoint and we call the vertices of $B^-$ and $B^+$ the \emph{boundary vertices} of $H$. It is a \emph{path piece} if $H$ has a path decomposition such that all vertices of $B^-$ are (also) in the first bag and all vertices of $B^+$ are (also) in the last bag. 

Now let $G$ be any directed graph or piece. We say that the piece $(H,B^-,B^+,a,b)$ is a \emph{valid} piece for $v \in V(G)$ if the in-degree of $v$ is $a$ and the out-degree of $v$ is $b$. Then the \emph{placement} of the valid piece for $v$ in $G$ replaces $v$ by $H$ such that the original incoming and outcoming arcs of $v$ are identified with the entry and exit arcs (respectively) of $H$ in any way that forms a bijection. This terminology enables the following convenient lemmas:

\begin{lemma} \label{lem:puzzle}
Let $G$ be a path piece with boundary vertices $B^-,B^+$. Let $S \subseteq V(G)$. Suppose that the assumed path decomposition $(\{X_i \mid i \in I\}, T=(I,F))$ of $G$ has width $w$ and, for every $v \in S$, all in-neighbors of $v$ also appear in the first bag containing $v$. Moreover, for every $v \in S$, let $(H_v,B^-_v,B^+_v,a_v,b_v)$ be a valid path piece for $v$ such that the assumed path decomposition of $H_v$ has width $w_v$. Let $G'$ be obtained from $G$ by the placement of the pieces of $H_v$ in $G$ for all $v \in S$. Then $G'$ is a path piece such that the required path decomposition has width at most the maximum over all $i \in I$ of:
$$|X_i \setminus S|  + \max_{v \in S \cap X_i} \left\{ w_v+1 + \sum_{v' \in (S \cap X_i) \setminus \{v\}} \max\{ |B^-_{v'}|, |B^+_{v'}| \}\right\} - 1.$$
\end{lemma}
\begin{proof}
We modify the path decomposition for $G$. We may assume that the path decomposition is such that for two consecutive bags $X_i, X_{i'}$, it holds that $|X_i \triangle X_{i'}| = 1$. For every $v \in S$, let $X_v$ be the first bag containing $v$ (note that $X_v$ will be the same for all $v$ in the first bag of the path decomposition of $G$, but will otherwise be unique). Then, for each $v \in S$, iteratively, replace $v$ in $X_v$ by $B^-_v$. Then, create a number of copies of the $i \in I$ corresponding to $X_v$, where the number of copies is equal to the number of nodes of the assumed path decomposition of $H_v$. To each of these new bags, add the vertices in the bags of the assumed path decomposition of $H_v$ in the natural order. We need special care again for the first bag of the path decomposition: here we expand the decomposition for one vertex after the other. Finally, we add $B^+_v$ to all further bags containing $v$. The claimed bound immediately follows from this construction. 
\end{proof}

We consider the following strengthening of Lemma~\ref{lem:puzzle}.

\begin{lemma} \label{lem:puzzle2}
Let $G$ be a path piece with boundary vertices $B^-$ and $B^+$. Let $S \subseteq V(G)$. Suppose that $G$ has a path decomposition $(\{X_i \mid i \in I\}, T=(I,F))$ of width $w$, the first bag contains at most one vertex from $S$, no bag contains more than two vertices from $S$, and, for every $v \in S$, every in-neighbor $u$ of $v$ is contained in the first bag containing $v$ and if $u \in S$, not in any subsequent bags. Moreover, for every $v \in S$, let $(H_v,B^-_v,B^+_v,a_v,b_v)$ be a valid path piece for $v$ such that the assumed path decomposition of $H_v$ has width $w_v$. Let $G'$ be obtained from $G$ by the placement of the pieces of $H_v$ in $G$ for all $v \in S$. Then $G'$ is a path piece such that the required path decomposition has width at most the maximum over all $i \in I$ of:
$$|X_i \setminus S|  + \max\left\{ \max_{v \in S \cap X_i} \left\{ w_v+1 \right\}, \sum_{v \in S \cap X_i} \max\{ |B^-_{v'}|, |B^+_{v'}| \}\right\} - 1.$$
\end{lemma}
\begin{proof}
We modify the path decomposition for $G$. We may assume that the path decomposition is such that for two consecutive bags $X_i, X_{i'}$, it holds that $|X_{i'} \setminus X_{i}| \leq 1$. For every $v \in S$, let $X_v$ be the first bag containing $v$ (note that $X_v$ is unique by the previous assumption and the assumption that the first bag contains at most one vertex from $S$). Treat the bags of the path decomposition in the path order on $T$. Consider the vertex $v \in S$ for which $X_v$ comes first in this order. Replace $v$ by $B^-_v$ in this bag. Then, create a number of copies of this bag equal to the number of bags of the assumed path decomposition of $H_v$, and insert these bags (joined with the vertices of $X_v \setminus\{v\}$) into the path decomposition after $X_v$. At the end, we have a bag containing $(X_v \setminus \{v\}) \cup B^+_v$, because $(H_v,B^+_v,B^-_v,a_v,b_v)$ is a path piece. We now continue along the path order of $T$ and replace $v$ by $B^+_{v}$ in every bag we encounter, until the first bag we encounter containing a vertex $v' \in S \setminus \{v\}$. In this bag, we still replace $v$ by $B^+_v$, then replace $v'$ by $B^-_{v'}$ and create a new subsequent bag where we remove $B^+_v$. This yields a bag containing $(X_{v'} \setminus \{v,v'\}) \cup B^-_v$. Since $v$ will not appear in any further bags by the assumptions of the lemma, this is safe. Then, we continue with the same treatment for $v'$ as we did before with $v$, etc., all thew way until we reach the end of $T$. The claimed bound immediately follows from this construction. 
\end{proof}

We can prove a similar lemma with respect to tree decompositions. A piece $(H,B^-,B^+,a,b)$ is a \emph{tree piece} if $H$ has a tree decomposition that has a bag containing $B^- \cup B^+$.

\begin{remark} \label{rem:puzzle:tw-pw}
Any path piece $(H,B^-,B^+,a,b)$ with an assumed path decomposition of width $w$ is also a tree piece with a tree decomposition of width at most $w + |B^+|$ by adding $B^+$ to every bag.
\end{remark}

\begin{lemma} \label{lem:puzzletw}
Let $G$ be a tree piece with boundary vertices $B^-,B^+$ such that the assumed tree decomposition $(\{X_i \mid i \in I\}, T=(I,F))$ has width $w$. Let $S \subseteq V(G)$. For every $v \in S$, let $(H_v,B^-_v,B^+_v,a_v,b_v)$ be a valid path piece for $v$ such that the assumed tree decomposition of $H_v$ has width $w_v$. Let $G'$ be obtained from $G$ by the placement of the pieces of $H_v$ in $G$ for all $v \in S$. Then $G'$ is a tree piece such that the required tree decomposition has width at most:
$$\max\left\{ \max_{v \in S} \{w_v\}, \max_{i\in I} \left\{|X_i \setminus S|  + \max_{v \in S \cap X_i} \{ |B^-_{v'} \cup B^+_{v'}| \} \right\} - 1\right \}.$$
\end{lemma}
\begin{proof}
We modify the tree decomposition for $G$. For each $v \in S$, replace $v$ in all bags containing $v$ by $B^-_v \cup B^+_v$. For each $v \in S$, add the assumed tree decomposition of $H_v$ to $T$ by adding an edge between any node of this tree decomposition and any node $i \in I$ for which $v$ used to be in $X_i$. The claimed bound immediately follows from this construction.
\end{proof}

\begin{remark} \label{rem:puzzle}
The same lemmas hold (with simple modifications) in the case of directed or undirected graphs instead of (directed) path/tree pieces.
\end{remark}

We now describe both gadgets in detail.

\paragraph{\chgad{a} Gadget}
Let $a$ be a positive integer. The \emph{\chgad{a} gadget} gadget will allow $a$ units of flow of commodity~2 through one direction, unless $1$ unit of flow of commodity~1 flows through the other direction.

\begin{figure}[tb]
    \centering
    \includegraphics[scale=0.8]{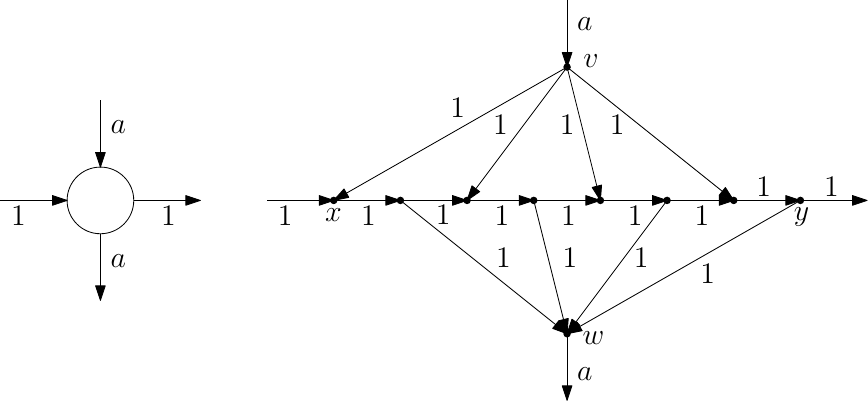}
    \caption{The \chgad{a} gadget. Left: the schematic representation of the gadget, with its entry and exit arcs. Right: the full construction for $a=4$, with arcs labelled by their capacities.}
    \label{figure:gadgetl}
\end{figure}

The gadget is constructed as follows. See Figure~\ref{figure:gadgetl} for its schematic representation and for an example with $a=4$. We build a directed path $P$ with $2\cdot a$ vertices. We add two additional vertices $v$ and $w$. The vertex $v$ has arcs towards the first, third, fifth, etc.\ vertices of the path, and $w$ has arcs from the second, fourth, sixth, etc.\ vertices of the path. All these arcs and the arcs of the path have capacity~$1$. We add an incoming arc of capacity $1$ to the leftmost vertex $x$ of the path and an incoming arc of capacity $a$ to $v$. We call these the \emph{entry arcs} of the gadget. We add an outgoing arc of capacity $1$ at the rightmost vertex $y$ of the path and an outgoing arc with capacity $a$ to $w$. We call these the \emph{exit arcs} of the gadget. 

We call $v,w,x,y$ the \emph{boundary vertices} of the gadget. Note that all arcs incoming to or outgoing from the gadget are incident on boundary vertices.

Observe that this gadget can be constructed in time polynomial in the given value of $a$ if $a$ is given in unary, as the gadget has size linear in $a$. We capture the functioning of the gadget in the following lemma.

\begin{lemma} \label{lem:chGad} 
    Consider the \chgad{a} gadget for some integer $a$. Let $f$ be some $2$-commodity flow such that the entry arc at $x$ and the exit arc at $y$ only carry flow of commodity~$1$ and such that the entry arc at $v$ and the exit arc at $w$ only carry flow of commodity~$2$. Then:
    \begin{itemize}
    \item If $v$ receives $a$ units of commodity~$2$, then $x$ receives no flow of commodity~$1$.
    \item If $x$ receives $1$ unit of commodity~$1$, then $v$ receives no flow of commodity~$2$.
    \end{itemize}
\end{lemma}
\begin{proof}
    Suppose that $v$ receives $a$ units of commodity~$2$. Then, every arc leaving $v$ is used to capacity by commodity~$2$. Since the exit arc at $y$ can only have flow of commodity~$1$, the flow of commodity~$2$ can only exit the gadget at $w$. This means that the arc leaving $x$ is used to capacity by commodity~$2$. As such, $x$ cannot receive any flow of commodity~$1$.

    Suppose that $x$ receives $1$~unit of commodity~$1$. Since $w$ can only receive flow of commodity~$2$, the flow of commodity~$1$ can only exit the gadget at $y$. This means that all the arcs in the path $P$ are used to capacity by commodity~$1$. As such, $w$ cannot receive any flow of commodity~$2$.
\end{proof}

\begin{lemma} \label{lem:gatepw}
For any integer $a$, the \chgad{a} gadget is a path piece such that the required path decomposition has width~$3$.
\end{lemma}
\begin{proof}
The gadget is a piece by construction, with $B^- = \{v,x\}$, $B^+ = \{w,y\}$, and $a=b=2$. To construct the path decomposition, begin with a path decomposition where the endpoints of each edge on the path from $x$ to $y$ are in a consecutive bags together. Add $v,w$ to every bag. This is a valid path decomposition of width~$3$.
\end{proof}

\paragraph{Verifier Gadget}
Let $L$ be a (typically large) integer. Let $S\subseteq [L-1]$ be a set of integers, where $S = \{a_1,\ldots,a_{|S|}\}$. An \emph{$(S,L)$-Verifier} gadget is used to verify that the amount of flow through an edge is in $S$.

\begin{figure}[tb]
    \centering
    \includegraphics[width=0.9\textwidth]{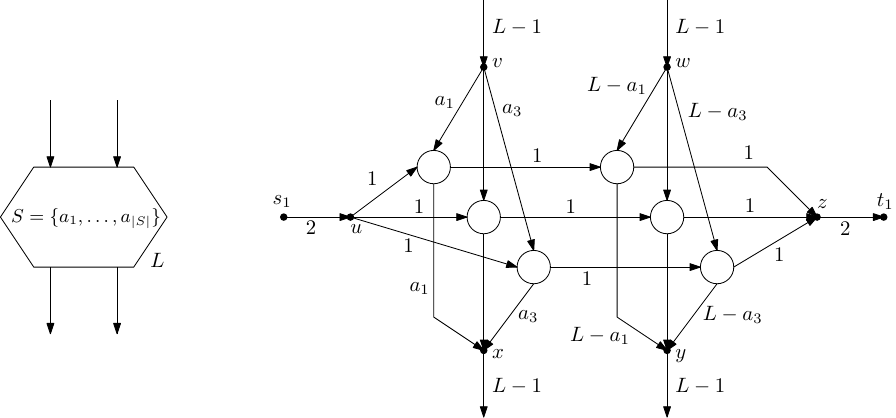}
    \caption{The $(S,L)$-Verifier gadget. Left: a schematic representation of the gadget, with its entry and exit arcs. The value on the bottom-right of the schematic representation denotes the sum of the incoming flows to the gadget. Right: the graph that realises the gadget for $S = \{a_1,a_2,a_3\}$, with arcs labelled by their capacities (the unlabelled arcs have capacities $a_2$ and $L-a_2$ respectively; their labels are omitted for clarity).}
    \label{figure:verifier}
\end{figure}

The gadget is constructed as follows. See Figure~\ref{figure:verifier} its schematic representation and for an example with $S=\{a_1,a_2,a_3\}$. 
We add six vertices $u$, $v$, $w$, $x$, $y$, $z$; these are the \emph{boundary vertices} of this gadget.
We add arcs $s_1u$ and $zt_1$ of capacity $r-1$. Then, $v$ and $w$ have incoming arcs of capacity $L-1$ (the \emph{entry arcs} of this gadget) and $x$ and $y$ have outgoing arcs of capacity $L-1$ (the \emph{exit arcs} of this gadget). 
Finally, we have $r$ rows of two \chgadn{} gadgets each. The $i$th row has two \chgad{a_i} gadgets with the following arcs:\begin{itemize}
    \item an arc of capacity $1$ from $u$ to the first \chgadn{} gadget,
    \item an arc of capacity $a_i$ from $v$ to the first \chgadn{} gadget,
    \item an arc of capacity $a_i$ from the first \chgadn{} gadget to $x$,
    \item an arc of capacity $1$ from the first to the second \chgadn{} gadget,
    \item an arc of capacity $L-a_i$ from $w$ to the second \chgadn{} gadget,
    \item an arc of capacity $L-a_i$ from the second \chgadn{} gadget to $y$,
    \item an arc of capacity $1$ from the second \chgadn{} gadget to $z$.
\end{itemize}   

Observe that this gadget can be constructed in time polynomial in the values of $a_1,\ldots,a_{|S|}$ if they are given in unary, as the Gate gadgets have size linear in the given value. We capture the functioning of the gadget in the following lemma.

\begin{lemma}\label{lemma:verifierlemma}
Consider the $(S,L)$-Verifier gadget for some integer $L$ and some $S \subseteq [L-1]$.
Let $f$ be some $2$-commodity flow such that $s_1$ sends ${|S|}-1$ units of flow of commodity~1 to $u$, $t_1$ receives ${|S|}-1$ units of flow of commodity~1 from $z$. Suppose there is an integer $\alpha \in [L-1]$ such that $v$ receives $\alpha$ units of flow of commodity~2 over its entry arc, $w$ receives $L-\alpha$ units of flow of commodity~2 over its entry arc, and neither entry arc receives flow of commodity~1.
Then:
\begin{itemize}
    \item $\alpha\in S$,
    \item $x$ sends $\alpha$ units of flow of commodity~2 over its exit arc,
    \item $y$ sends $L-\alpha$ units of flow of commodity~2 over its exit arc.
\end{itemize}
Conversely, if $\alpha \in S$, then there exists a $2$-commodity flow that fulfills the conditions of the lemma.
\end{lemma}

\begin{proof}
    Since $u$ receives ${|S|}-1$ units of flow of commodity~1 and has ${|S|}$ outgoing arcs, $u$ sends $1$ unit of flow of commodity~1 over all but one of its outgoing arcs. Recall that $v$ and $w$ do not receive flow of commodity~1. If $1$ unit of flow of commodity~1 is sent over say the $i$th outgoing arc of $u$, in order to arrive at $z$, it must go through the \chgad{a_i} gadget and the corresponding \chgad{(L-a_i)} gadget. The same amount of flow must also leave these gadgets and thus by Lemma \ref{lem:chGad}, these gadgets cannot transfer flow of commodity~2. Thus, the flow of commodity~2 that $v$ and $w$ receive must go through the \chgadn{} gadgets of the row where $u$ has not sent any flow to, say this is row $j$. $x$ and $y$ must send flow through arcs of capacity $a_j$ and $L-a_j$, respectively. The capacities of the arcs and the gadgets enforce that $\alpha \leq a_j$ and $L-\alpha \leq L - a_j$, respectively. Thus, $\alpha = a_j$, and so $\alpha \in S$.
    
    All the flow that the \chgad{(L-a_j)} gadget receives (which is only of commodity~2) must be sent to $y$, since $t_1$ is a sink. Hence, $y$ sends $L-\alpha$ flow through its exit arc. The flow that the \chgad{a_j} gadget receives (which is only of commodity~2) must all be sent to $x$: we cannot send even $1$ unit of flow over the horizontal arc to the second gadget, as the second gadget would then receive $L-\alpha+1$ units of flow of commodity~2. However, it cannot send flow of commodity~2 to $z$ (as $z$ only has an arc to $t_1$) and can send out at most $L-\alpha$ flow to $y$. Thus, $x$ sends $\alpha$ units of flow over its exit arc.

    For the converse, $u$ sends $1$ unit of flow of commodity~1 through all \chgadn{} gadgets of the values unequal to $\alpha$, and $\alpha$ and $L-\alpha$ flow through the two other \chgadn{} gadgets.
\end{proof}

We note that in the later proofs, for a particular choice of $L$, we may also use Verifier gadgets with capacities $2L-1$, and incoming and outgoing flows adding up to $2L$ instead of $L$. We will explicitly indicate that, also in the schematic representation.

\begin{lemma} \label{lem:verifier_pathwidth}
For any integer $L$ and any $S \subseteq [L-1]$, the $(S,L)$-Verifier gadget is a path piece such that the required path decomposition (ignoring $s_1$ and $t_1$) has width~$9$.
\end{lemma}
\begin{proof}
The gadget is a piece by construction, with $B^- = \{u,v,w\}$ and $B^{+} = \{x,y,z\}$. If we treat the \chgadn{} gadgets as vertices, then we can immediately construct a path decomposition of width~$7$, by putting $u$, $v$, $w$, $x$, $y$, $z$ in all bags and adding the two vertices corresponding to each row of \chgadn{} gadgets in successive bags. Since the \chgadn{} gadgets are each path pieces with a decomposition of width~$3$ by Lemma~\ref{lem:gatepw} and each have two entry and exit arcs, we obtain a path decomposition of the whole Verifier of width~$9$ by Lemma~\ref{lem:puzzle2}.
\end{proof}

\subsubsection{Reductions for Directed Graphs}
Using the gadgets we just proposed, we show our hardness results for \ILCF{2} (i.e.~the case of directed graphs) with pathwidth and with treewidth as parameter.

We note that our hardness construction will be built using only Verifier gadgets as subgadgets. The entry arcs and exit arcs of this gadget are meant to transport solely flow of commodity~2 per Lemma~\ref{lemma:verifierlemma}. Hence, in the remainder, it helps to think of only commodity~2 being transported along the edges, so that we may focus on the exact value of that flow to indicate which vertex is selected or whether two selected vertices are adjacent. We later make this more formal when we prove the correctness of the reduction.

\twocomXNLP*
\begin{proof}
Membership in XNLP can be seen as follows. Take a path decomposition of $G$, say with successive bags $(X_1, \ldots, X_r)$. One can build a dynamic programming table, where each entry is indexed by a node $j$ with associated bag $X_j$ and a function $f^i_j : X_j \rightarrow [-C,C]$, where $C$ is some upper bound on the maximum value of the flow of any commodity (note that $C$ is linear in the input size), $i = 1,2$ and $1\leq j \leq r$. One should interpret $f^i_j$ as mapping each vertex $v \in X_j$ to the net difference of flow of commodity $i$ in- or outgoing on that vertex in a partial solution up to bag $X_j$. The content of the table is a Boolean representing whether
there is a partial flow satisfying the requirements that $f^i_j$ sets. Basic application of dynamic programming on (nice) path decompositions can solve the \ILCF{2} problem with this table. This dynamic programming algorithm
can be transformed to a non-deterministic algorithm by not building entire tables, but
instead guessing one element for each table
with positive Boolean. The guessed element of
the table can be represented by $O(\pw(G)\log C)$ bits; in addition, we need $O(\log n)$ bits to know which bag of the path decomposition we
are handling, and to look up relevant information of the graph. This yields an non-deterministic
algorithm with
$O(\pw(G) \log C + \log n)$ memory.

For the hardness, we use a reduction from {\sc Chained Multicolour Clique} (see Theorem~\ref{thm:cmc-tcmc}). Suppose we have an instance of {\sc Chained Multicolour Clique}, with a graph $G=(V,E)$,
colouring $c:V\rightarrow [k]$, and partition $V_1, \ldots, V_r$ of $V$.

Build a Sidon set with $|V|$ numbers by applying the algorithm of Theorem~\ref{thm:sidon}. Following the same theorem, the numbers are in $[4|V|^2]$. Set $L=4|V|^2+1$ to be a `large' integer. To each vertex $v \in V$, we assign a unique element of the set $S$, denoted by $S(v)$. For any subset $V'\subseteq V$, let $S(V') = \{S(v) \mid v \in W\}$. For any subset $E'\subseteq E$, let $S(E') = \{S(u)+S(v)\mid uv \in E'\}$.

\begin{figure}
    \centering
    \includegraphics{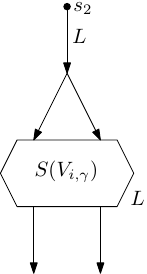}
    \caption{Vertex selector gadget used in the proof Theorem~\ref{thm:2comXNLP}. The flow sent through the $(V_{i,\gamma},L)$-Verifier will correspond to selecting a vertex from the class $V_{i,\gamma}$.}
    \label{figure:vertexselection}
\end{figure}

We now describe several (further) gadgets that we use to build the full construction.
Let $V_{i,\gamma}$ be the vertices in $V_i$ with colour $\gamma$. Each set $V_{i,\gamma}$ is called
a {\em class}. For each class $V_{i,\gamma}$, we use a \emph{Vertex selector gadget} to select the vertex from $V_{i,\gamma}$ that should be in the solution to the {\sc Chained Multicoloured Clique} instance. This gadget (see Figure~\ref{figure:vertexselection}) consists of a single $(V_{i,\gamma}, L)$-Verifier gadget, where its entry arcs jointly start in a single vertex that in turn has a single arc from $s_2$ of capacity $L$. Intuitively, we select some $v \in V_{i,\gamma}$ if and only if the left branch receives $S(v)$ flow and the right branch receives $L-S(v)$ flow.

\begin{figure}
    \centering
    \includegraphics{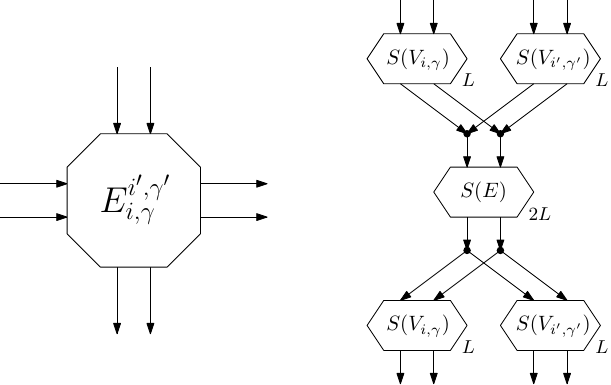}
    \caption{Edge check gadget for the combination $V_{i,\gamma}$ and $V_{i',\gamma'}$. Left: the schematic representation of the gadget. The flow corresponding to the $V_{i,\gamma}$ vertex enters on the left and exits on the right. The flow corresponding to the $V_{i',\gamma'}$ vertex enters from the top and exits at the bottom. Right: the realisation of the gadget. The flow sent through the $(S(E),2L)$-Verifier must correspond to a unique sum of two incident vertices from the classes $V_{i,\gamma}$ and $V_{i',\gamma'}$. The $(S(V_{i,\gamma}),L)$-Verifiers and $(S(V_{i',\gamma'}),L)$-Verifiers ensure the right combination is checked, and the outgoing flow is the same as the incoming flow. Chaining Edge check gadgets makes some $(S(V_{i,\gamma}),L)$-Verifiers redundant, but this does not matter for the argument.}
    \label{figure:edgecheck}
\end{figure}

For each pair of incident classes, we construct an \emph{Edge check gadget}.
That is, we have an Edge check gadget for all classes $V_{i,\gamma}$ and $V_{i',\gamma'}$ with
$|i-i'|\leq 1$, and $\{i,\gamma\}\neq \{i',\gamma'\}$. 
Let $E^{i',\gamma'}_{i,\gamma} \subseteq E$ denote the set of edges between $V_{i,\gamma}$ and $V_{i',\gamma'}$. 
An Edge check gadget will check if two incident classes have vertices selected that are adjacent. 
The gadget (see Figure~\ref{figure:edgecheck}) consists of a central $(S(E), 2L)$-Verifier gadget (note that we could also use an $(S(E^{i',\gamma'}_{i,\gamma}),2L)$-Verifier gadget instead, but this is not necessary). Its entry arcs originate from two vertices that have as incoming arcs the exit arcs of an $(S(V_{i,\gamma}), L)$- and $(S(V_{i',\gamma'}), L)$-Verifier gadget. Its exit arcs head to two vertices that have as outgoing arcs the entry arcs of a different $(S(V_{i,\gamma}), L)$- and $(S(V_{i',\gamma'}), L)$-Verifier gadget. The gadget thus has four entry arcs and four exit arcs, corresponding to the entry arcs of the first two Verifier gadgets and the exit arcs of the last two Verifier gadgets.

Intuitively, if the entry arcs have flow (of commodity~2) of value
$S(v)$, $L-S(v)$, $S(w)$, and $L-S(w)$ consecutively (refer to Figure~\ref{figure:edgecheck}), then there is a valid flow if and only if $vw\in E^{i',\gamma'}_{i,\gamma}$. Note that the sum $S(v) + S(w)$ is unique, because $S$ is a Sidon set, and thus so is $2L - (S(v) + S(w))$. Hence, the only way for the flow to split up again and leave via the exit arcs is to split into $S(v)$, $S(w)$, $L - S(v)$, and $L - S(w)$; otherwise, it cannot pass the $(S(V_{i,\gamma}),L)$-Verifier or the $(S(V_{i',\gamma'}),L)$-Verifier at the bottom of the Edge check gadget. Hence, the exit arcs again have flow of values $S(v)$, $L-S(v)$, $S(w)$, and $L-S(w)$ consecutively, just like the entry arcs.

\medskip
With these gadgets in hand, we now describe the global structure of the reduction. Throughout, it will be more helpful to look at the provided figures than to follow the formal description.
For each class $V_{i,\gamma}$, we first create a Vertex selector gadget (as in Figure~\ref{figure:vertexselection}). 

\begin{figure}[tbp]
    \centering
    \includegraphics{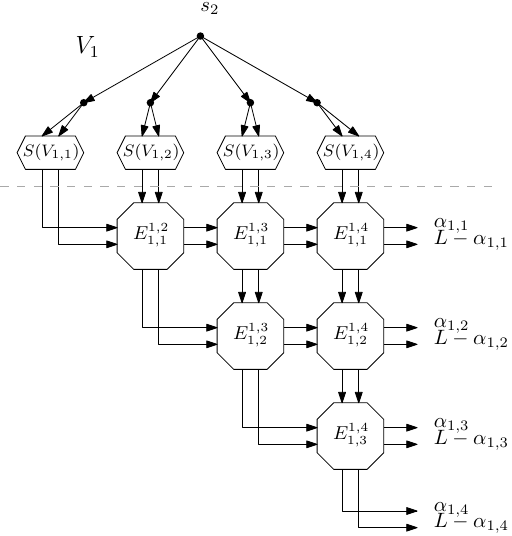}
    
    \caption{Illustration of the structure of Vertex selector and Edge check gadgets to enforce that the selected vertices from $V_1$ form a clique. In this example, $k=4$. The $\alpha_{i,\gamma}$'s denote the amount of flow, having a value of $S(v)$ for some $v \in V_{i,\gamma}$, selected in the corresponding $(S(V_{i,\gamma}),L)$-Verifiers of the Vertex selector gadgets.}
    \label{fig:XNLP PW Structure V1}
\end{figure}

We then create Edge check gadgets to check, for any $i \in [r]$, that the selected vertices in $V_{i,\gamma}$ for all $\gamma \in [k]$ are adjacent in $G$. We create $k$ rows of gadgets. Row $\gamma \in [k]$ has $k-\gamma$ Edge check gadgets, which correspond to checking that the vertices selected in $V_{i,\gamma}$ and $V_{i,\gamma'}$ are indeed adjacent (via edges in $E^{i,\gamma'}_{i,\gamma}$) for each $\gamma' \in [\gamma+1,k]$. 
The construction is as follows (see Figure~\ref{fig:XNLP PW Structure V1}). 
For any $\gamma \in [k], \gamma' \in [\gamma+1,k]$, the Edge check gadget for $E^{i,\gamma'}_{i,\gamma}$ has its left entry arcs unified with the bottom exit arcs of the Edge check gadget for $E^{i,\gamma'-1}_{i,\gamma-1}$ if $\gamma > 1$ and $\gamma' = \gamma+1$ and with the exit arcs of the Edge check gadget for $E^{i,\gamma'-1}_{i,\gamma}$ if $\gamma > 1$ and $\gamma' > \gamma + 1$.
The Edge check gadget for $E^{i,\gamma'}_{i,\gamma}$ has its top entry arcs unified with the bottom exit arcs of the Edge check gadget for $E^{i,\gamma'}_{i,\gamma-1}$ if $\gamma > 1$.

We call this the \emph{Triangle gadget} for $V_i$. Note that it has $2k$ entry arcs and $2k$ exit arcs. The former can be interpreted to be partitioned into $k$ columns, while the latter can be seen to be partitioned into $k$ rows (see Figure~\ref{fig:XNLP PW Structure V1}).

Note that chaining Edge check gadgets makes some $(S(V_{i,\gamma}),L)$-Verifiers redundant, as two such verifiers of the same type will follow each other. However, for the sake of clarity and as it does not matter for the overall, we leave such redundancies in place.

\begin{figure}[tb]
    \centering
    \includegraphics{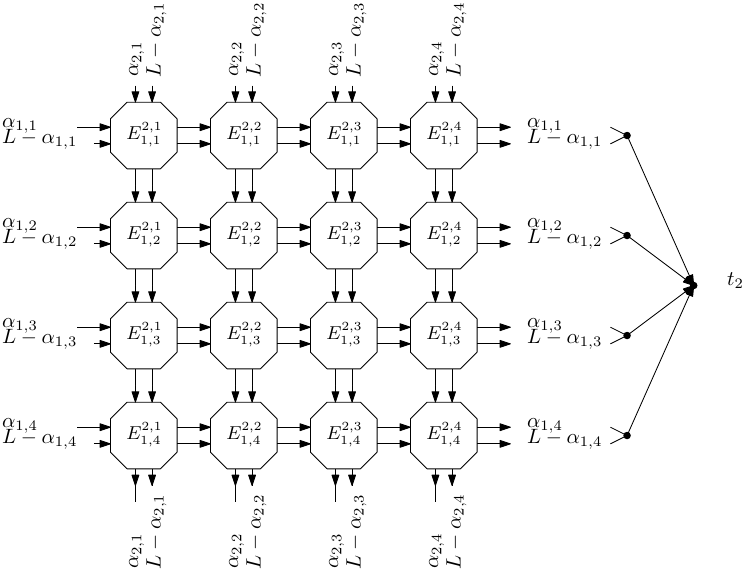}
    
    \caption{Illustration of the structure of the Edge check gadgets to enforce that the selected vertices from $V_1$ and from $V_2$ are complete to each other. In this example, $k=4$. The $\alpha_{i,\gamma}$'s denote the amount of flow, having a value of $S(v)$ for some $v \in V_{i,\gamma}$, selected in the corresponding $(S(V_{i,\gamma}),L)$-Verifiers.}
    \label{fig:XNLP PW Structure V1V2}
\end{figure}

Next, we create Edge check gadgets to check, for any $i \in [r-1]$, that the selected vertices in $V_{i,\gamma}$ and $V_{i+1,\gamma'}$ for all $\gamma,\gamma' \in [k]$ are indeed adjacent in $G$. We create a grid with $k$ rows and $k$ columns of gadgets, where each row corresponds to a combination of $\gamma$ and $\gamma'$. The construction is as follows (see Figure~\ref{fig:XNLP PW Structure V1V2}). For each $\gamma,\gamma' \in [k]$, the Edge check gadget for $E^{i+1,\gamma'}_{i,\gamma}$ has its left entry arcs identified with the right exit arcs of the Edge check gadget for $E^{i+1,\gamma'-1}_{i,\gamma}$ if $\gamma' > 1$ and it has its top entry arcs identified with the bottom exit arcs of the Edge check gadget for $E^{i+1,\gamma'}_{i,\gamma-1}$ if $\gamma > 1$. 

We call this the \emph{Square gadget} for $V_i$ and $V_{i+1}$. Note that it has $4k$ entry arcs and $4k$ exit arcs. The horizontal entry and exit arcs (see Figure~\ref{fig:XNLP PW Structure V1V2}) correspond to $V_{i}$, whereas the vertical entry and exit arcs correspond to $V_{i+1}$.

Then, we connect the Vertex selector and Triangle gadgets. For each $V_i$, we connect the Vertex selector gadget for $V_i$ to the Triangle gadget for $V_i$ as follows (see Figure~\ref{fig:XNLP PW Structure V1}). The Edge check gadget for $E^{i,2}_{i,1}$ of the Triangle gadget has its left entry arcs unified with the exit arcs of the Vertex selector gadget of $V_{i,1}$. For all $\gamma \in [1,k]$, the Edge check gadget for $E^{i,\gamma}_{i,\gamma-1}$ has its top entry arcs unified with the exit arcs of the Vertex selector gadget of $V_{i,\gamma}$.

Finally, we connect the Triangle and Square gadgets (see Figure~\ref{fig:XNLP PW Structure full}). The Square gadget for $V_1$ and $V_2$ has its horizontal entry arcs identified with the exit arcs of the Triangle gadget for $V_1$. Then, for the Square gadget for $V_i$ and $V_{i+1}$ for any $i \in [r-1]$, it has its vertical entry arcs identified with the exit arcs of the Triangle gadget for $V_{i+1}$. Its $2k$ horizontal exit arcs are paired (one pair per colour class) and directed to a vertex; these $k$ vertices are then each connected by a single arc of capacity $L$ to $t_2$ (cf.~Figure~\ref{fig:XNLP PW Structure V1V2}). Finally, we also do the latter for the vertical exit arcs of the Square gadget for $V_{r-1}$ and $V_{r}$. 

We now set the demand for commodity~1 to the sum of the capacities of the outgoing arcs of $s_1$ (which is equal to the sum of the capacities of the incoming arcs of $t_1$). We set the demand for commodity~2 to the sum of the capacities of the outgoing arcs of $s_2$ (which is equal to the sum of the capacities of the incoming arcs of $t_2$). This completes the construction. We now prove the bound on the pathwidth, followed by the correctness of the reduction and a discussion of the time and space needed to produce it. To prove the pathwidth bound, we note that all constructed gadgets are path pieces, and thus we can apply Lemma~\ref{lem:puzzle}.

\begin{figure}[tbp]
    \centering
    \includegraphics[scale=.75]{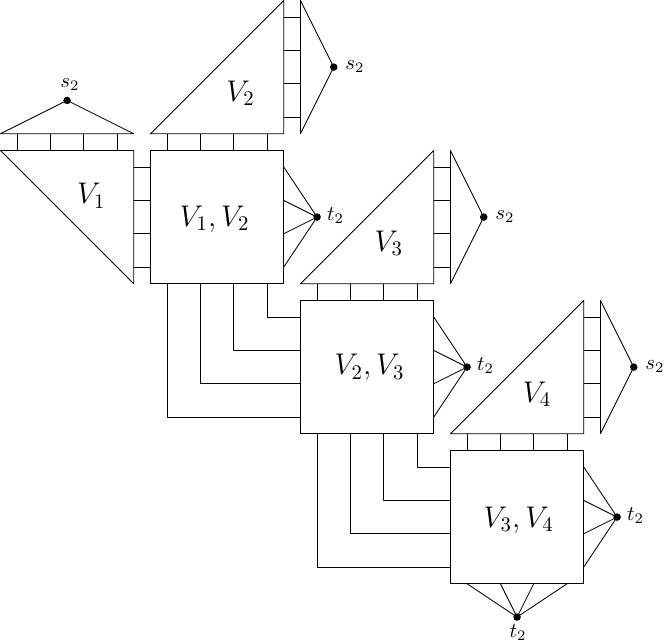}
    
    \caption{Overview of the complete structure of the reduction for $r=4$. Triangles represent a structure as in Figure~\ref{fig:XNLP PW Structure V1}, and squares a structure as in Figure~\ref{fig:XNLP PW Structure V1V2}. Directions are not drawn, but clear from Figure~\ref{fig:XNLP PW Structure V1} and~\ref{fig:XNLP PW Structure V1V2}. The labels inside each block (say $V_i$ or $V_i,V_{i+1}$) denote that flow corresponding to vertices of this set (i.e.\ $V_i$ or $V_i$ and $V_{i+1}$) is flowing in a block. Note that all points labelled $s_2, t_2$ are indeed the same vertex.}
    \label{fig:XNLP PW Structure full}
\end{figure}

\begin{claim} \label{clm:xnlp-pw-bound}
The constructed graph has pathwidth at most $8k+O(1)$.
\end{claim}
\begin{claimproof}
We construct a path decomposition as follows. First, we will ensure that $s_1,t_1,s_2,t_2$ are in every bag. Then, we make the following observations about the gadgets we construct. Notice that Vertex selector gadgets and Edge check gadgets have pathwidth $O(1)$, which directly follows from Lemma~\ref{lem:puzzle} and the fact that Verifier gadgets have pathwidth $O(1)$, as proven in Lemma~\ref{lem:verifier_pathwidth}. Note that Vertex selector gadgets have two exit arcs, whereas Edge check gadgets have four entry arcs and four exit arcs.

The Triangle gadget is a subgraph of a $k \times k$ grid, while the Square gadget is a $k \times k$ grid. Note that a standard path decomposition of the grid has width $k$ and satisfies the conditions of Lemma~\ref{lem:puzzle}. Hence, following Lemma~\ref{lem:puzzle} and the above bounds for the Vertex selector and Edge check gadgets, the pathwidth of the Triangle gadget and the Square gadget is at most $4k+O(1)$. The Triangle gadget has $2k$ entry arcs and $2k$ exit arcs, whereas the Square gadget has $4k$ entry arcs and $4k$ exit arcs.

Finally, the full construction (treating Triangle and Square gadgets as vertices) has a path decomposition of width~$2$, as it is a caterpillar (see Figure~\ref{fig:XNLP PW Structure full}), with at most two vertices corresponding to Triangle or Square gadgets per bag. Applying Lemma~\ref{lem:puzzle}, we obtain a bound of $8k+O(1)$.
\end{claimproof}
We note that a slightly stronger bound of $4k+O(1)$ seems possible with a more refined analysis, but this bound will be sufficient for our purposes.

\begin{claim}
The given \textsc{Chained Multicolour Clique} instance has a solution if and only if the constructed instance of \ILCF{2} has a solution.
\end{claim}
\begin{claimproof}
For the forward direction, assume there exists a chained multicolour clique~$W$ in $G$. We construct a flow. We first consider commodity~2. Recall that one vertex is picked per $V_{i,\gamma}$ class by definition and thus $W$ has size $rk$. If $v\in V_{i,\gamma} \cap W$ for some $i \in [r], \gamma \in [k]$, then in the Vertex selector gadget of $V_{i,\gamma}$, we send $S(v)$ units of flow of commodity~2 to the left and $L-S(v)$ units of flow to the right into the Verifier gadget (see Figure~\ref{figure:vertexselection}). In any Verifier gadget, we route the flow so that it takes the path with capacity equal to the flow (see Lemma~\ref{lemma:verifierlemma}). This flow is then routed through all Edge check gadgets of the Triangle and Square gadgets, in the manner presented above in the description of Edge check gadgets. Since $W$ is a chained multicolour clique, the corresponding edge exist in $E$ and the flow can indeed pass through the $(S(E),2L)$-Verifier gadget of each Edge check gadget (again, see Lemma~\ref{figure:verifier}). For any $i \in [r-1]$, after passing through the Square gadget for $V_i$ and $V_{i+1}$, the flow originating in the Vertex selector gadget corresponding to $V_i$ is sent to $t_2$ through the horizontal exit arcs of the Square gadget (see Figure~\ref{fig:XNLP PW Structure V1V2}). Finally, the flow originating in the Vertex selector gadget corresponding to $V_r$ is sent to $t_2$ through the vertical exit arcs of the Square gadget. Since $W$ has size equal to $rk$ and we send $L$ units of flow through each Vertex selector gadget, all arcs from $s_2$ and to $t_2$ are used to capacity.

Next, we consider commodity~1. All flow of commodity~1 is routed through the unused gates in the Verifier gadgets, which is possible as we only use one vertical path per gadget for the flow corresponding to a vertex or edge (see Figure~\ref{figure:verifier}). It follows that we use all arcs from $s_1$ and to $t_1$ to capacity.

For the other direction, suppose we have an integer $2$-commodity flow~$f$ that meets the demands. Hence, there is a $2$-commodity flow in the constructed graph with all arcs from $s_1$ and $s_2$ and to $t_1$ and $t_2$ used to capacity, meaning that their total flow is equal to their total capacity, with flow with the corresponding commodity: commodity~1 for $s_1$ and $t_1$, and commodity~2 for $s_2$ and $t_2$. 

We first reason that the flow of commodity~1 behaves as expected in the Verifier gadgets. Notice that the constructed graph is a directed acyclic graph. Abstractly, it can also be seen as a directed acyclic graph where the vertices are Verifier gadgets. Hence, there exists a topological ordering on the Verifier gadgets. We refer to Figure~\ref{figure:verifier} to recall the naming of the vertices $u$ and $z$ of a Verifier gadget. We prove by induction on the topological ordering that the flow of commodity~1 in $f$ that enters over the arc $s_1u$ leaves over the arc $zt_1$, while staying in the gadget for the entirety of its flow path. As the base case, consider the last Verifier gadget $H$ in the ordering. The bottom of this gadget (refer to Figure~\ref{figure:verifier}) has arcs only going to $t_2$, by construction and the fact that it is the last Verifier gadget in the ordering. Hence, flow from $u\in H$ must all go to $z\in H$ and the flow on the arcs $s_1u$ and $zt_1$ fully fills the capacities. Then, indeed, flow of commodity~1 behaves as in the precondition of Lemma~\ref{lemma:verifierlemma}, and this flow does not leave the gadget downwards. For the induction step, consider some Verifier gadget $H$ in the ordering. By the induction hypothesis, all Verifier gadgets later in the ordering have the arcs to $t_1$ fully filled. But then the flow of commodity~1 in $H$ can only go to $z\in H$ to go to $t_1$. We get that the flow on the arcs $s_1u$ and $zt_1$ in $H$ fully fills the capacities, and does not leave the gadget downwards. Hence, flow of commodity~1 behaves as in the precondition of Lemma~\ref{lemma:verifierlemma}. By induction, all flow of commodity~1 behaves `properly' for Lemma~\ref{lemma:verifierlemma}.

By applying Lemma~\ref{lemma:verifierlemma} to every Verifier gadget, we get that the amount flow of commodity~2 always corresponds to some $\alpha \in S$ for the associated set $S$ of the gadget, and the left and right exit arcs carry $\alpha$ and $L-\alpha$ units of flow of commodity~2 respectively.
Now, the arc from $s_2$ in a Vertex selector gadget for $V_{i,\gamma}$
must have $L$ flow of commodity~2 and this must be split in $\alpha$ and $L-\alpha$, with $\alpha\in S(V_{i,\gamma})$. Thus, there is a $v\in V_{i,\gamma}$ with $\alpha=S(v)$, which corresponds to placing $v$ in the
chained multicoloured clique.
In any Edge check gadget, the flow of $\alpha\in S(V_{i,\gamma})$, $L-\alpha$ and $\beta\in S(V_{i',\gamma'})$,$L-\beta$ combines to a unique sum $\alpha+\beta$ and $2L - (\alpha + \beta)$, and assures that the edge between the corresponding vertices is present. As reasoned before, the flow must split back up into $\alpha$, $L-\alpha$ and $\beta,L-\beta$ by the unique sum due to the fact that $S$ is a Sidon set.
We get that the chosen vertices indeed form a chained multicolour clique, as the Edge check gadgets in the Triangle and Square gadgets enforce that all edges between selected vertices are present as should be.
\end{claimproof}

Finally, we claim that the constructed graph with its 
capacities can be built using $O(f(k)+\log n)$ space, for some computable function $f$. First, Sidon sets can be built in logarithmic space (cf.\ Theorem~\ref{thm:sidon}). Second, we use the (for log-space reductions standard) technique of not storing intermediate results, but recomputing parts whenever needed. E.g., each time we need the $i$th number of the Sidon set, we recompute the set and keep the $i$th number. Viewing the computation as a recursive program, we have constant recursion depth, while each level uses $O(f(k)+\log n)$ space, for some computable function $f$. The result now follows.
\end{proof}

To show the XALP-hardness of the \ILCF{2} parameterised by treewidth, we reduce from {\sc Tree-Chained Multicolour Clique} in a similar, but more involved manner.

\twocomXALP*
\begin{proof}
Membership in XALP can be seen as follows. Take a tree decomposition of $G$, say $(\{X_i~|~i\in I\}, T)$; assume that $T$ is binary. One can build a dynamic programming table, where each entry is indexed by a node $j \in I$ with associated bag $X_j$ and a function $f^i_j : X_j \rightarrow [-C,C]$, where $C$ is some upper bound on the maximum flow (note that $C$ is linear in the input size), $i = 1,2$ and $1\leq j \leq r$. One should interpret $f^i_j$ as mapping each vertex $v \in X_j$ to the net difference of flow of commodity $i$ in- or outgoing on that vertex in a partial solution in the subtree with bag $X_j$ as root. The content of the table is a boolean representing when there is a partial flow satisfying the requirements that $f^i_j$ sets. Basic application of dynamic programming on (nice) tree decompositions can solve the \ILCF{2} problem with this table.

Now, we transform this to a non-deterministic algorithm with a stack as follows. We basically guess an entry of each table, similar to the path decomposition case of Theorem~\ref{thm:2comXNLP}, and use the stack to handle nodes with two children.
Recursively, we traverse the tree~$T$. For a leaf bag, guess an entry of the table. For a bag with one child, guess an entry of the table given the guessed entry of the child. For a bag with two children, recursively get a guessed entry of the table of the left child. Put that entry on the stack. Then, recursively get a guessed entry of the table of the right child. Get the top element of the stack, and combine the two guessed entries.
We need $O(\log n)$ bits to store the position in $T$ we currently are and to look up information on $G$. We need $O(f(k)\log C)$ bits to denote a table entry and at each point, we have $O(1)$ such
table entries in memory. This shows membership in XALP.

For the hardness, we use a reduction from \textsc{Tree Chained Multicolour Clique} (see Theorem~\ref{thm:cmc-tcmc}). Suppose we have an instance of {\sc Tree Chained Multicolour Clique}, with a graph $G=(V,E)$, a tree partition $(\{V_i \mid i \in I\}, T=(I,F))$ with $T$ a tree of maximum degree~$3$, and a function $c \colon V \to [k]$.
    
As in the proof of Theorem~\ref{thm:2comXNLP}, we build a Sidon set $S$ of $|V|$ integers in the interval $[4|V|^2]$ using Theorem~\ref{thm:sidon}. To each vertex $v \in V$, we assign a unique element of the set $S$, denoted by $S(v)$. For any subset $W\subseteq V$, let $S(W) = \{S(v) \mid v \in W\}$ and for any subset $E'\subseteq E$, let $S(E') = \{S(u)+S(v)\mid uv \in E'\}$. Also, let $L= 4|V|^2+1$ be a `large' integer.

We now create a flow network, similar to the network in the proof of Theorem~\ref{thm:2comXNLP}. In particular, we recall the different gadgets that we created in that construction: the Vertex selector gadget (see Figure~\ref{figure:vertexselection}), the Edge check gadget (see Figure~\ref{figure:edgecheck}), the Triangle gadget (see Figure~\ref{fig:XNLP PW Structure V1}), and the Square gadget (see Figure~\ref{fig:XNLP PW Structure V1V2}). Their structure, functionalities, and properties will be exactly the same as before.

Root the tree $T$ at an arbitrary leaf~$r$. We may assume that $T$ has at least two nodes. We now create gadgets in the following manner, essentially mimicking the structure of $T$ (see Figure~\ref{fig:XALP TW Structure full}). Start with $r$. Let $r'$ be its child in $T$. Create a Triangle gadget for $V_r$ and a Square gadget for $V_r$ and $V_{r'}$. Identify the exit arcs of the Triangle gadget with the horizontal entry arcs of the Square gadget.

\begin{figure}[tb]
    \centering
    \includegraphics[scale=.7]{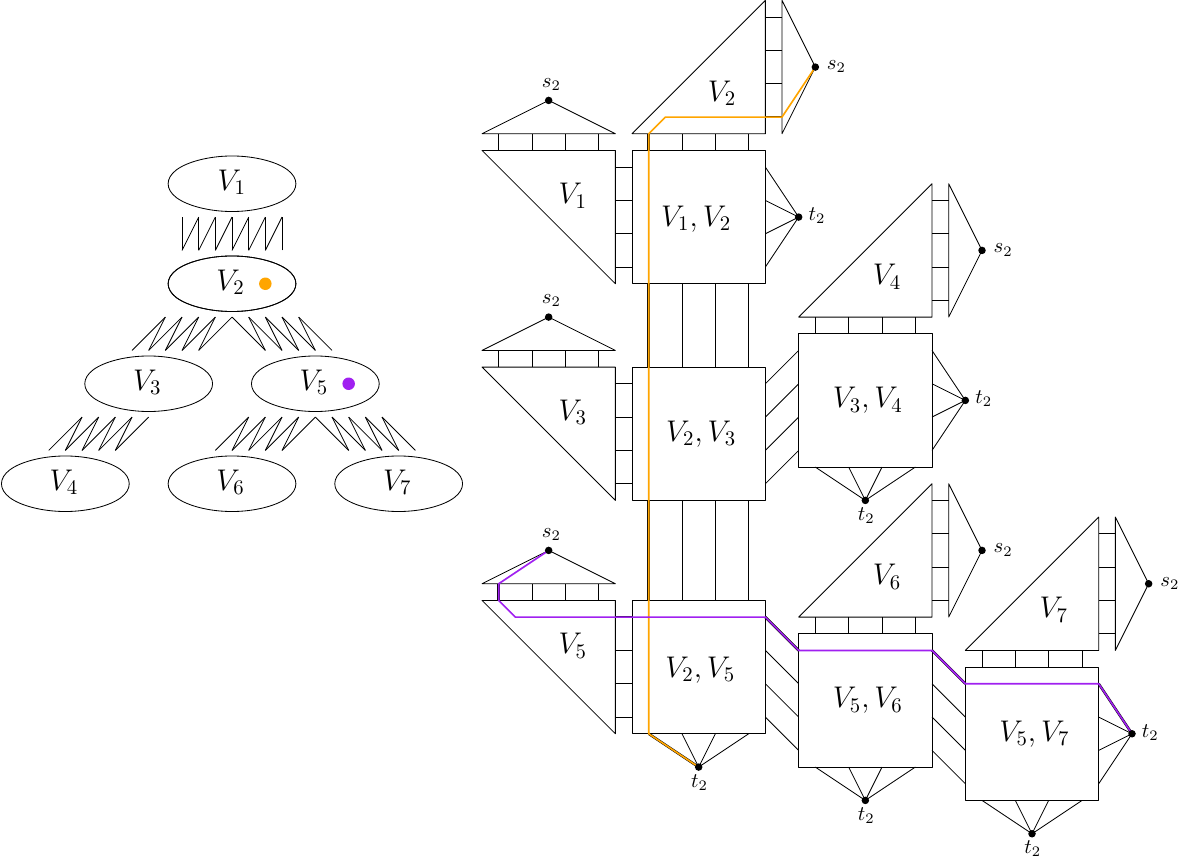}

    \caption{Overview of the complete structure of the reduction for $|I|=7$. Left: the structure of the input tree partition. Right: the structure of the reduction. Triangles represent a structure as in Figure~\ref{fig:XNLP PW Structure V1}, and squares a structure as in Figure~\ref{fig:XNLP PW Structure V1V2}. Directions are not drawn, but clear from Figure~\ref{fig:XNLP PW Structure V1} and~\ref{fig:XNLP PW Structure V1V2}. The labels inside each block (say $V_i$ or $V_i,V_{i+1}$) denote that flow corresponding to vertices of this set (i.e.\ $V_i$ or $V_i$ and $V_{i+1}$) is flowing in a block. Note that all points labelled $s_2, t_2$ are indeed the same vertex. Flow paths corresponding to a selected vertex in $V_2$ (orange) and one in $V_5$ (purple) are drawn as an example.}
    \label{fig:XALP TW Structure full}
\end{figure}

Now, for every node $i$ of $T$, starting with $r'$ and traversing the tree down in DFS order, do the following. Let $p$ be the parent of $i$. Suppose the depth of $p$ is odd. Create a Triangle gadget for $V_i$. Identify the exit arcs of this new Triangle gadget with the horizontal entry arcs of the (already constructed) Square gadget for $V_p$ and $V_i$. If $i$ has a child in $T$, let $d$ be the first child of $i$ in DFS order. Create a Square gadget for $V_i$ and $V_d$ and identify the horizontal exit arcs of the Square gadget of $V_p$ and $V_i$ with the horizontal entry arcs of the Square gadget of $V_i$ and $V_d$. If $i$ does not have a second child, then the $2k$ horizontal exit arcs of the Square gadget of $V_i$ and $V_d$ are paired (one pair per colour class) and directed to a vertex; these $k$ vertices are then each connected by a single arc of capacity $L$ to $t_2$ (cf.\ Figure~\ref{fig:XNLP PW Structure V1V2}). If $i$ has a second child $d'$, create a Square gadget for $V_i$ and $V_{d'}$, identify the horizontal exit arcs of the Square gadget of $V_i$ and $V_d$ with the horizontal entry arcs of the Square gadget of $V_i$ and $V_{d'}$. Then the $2k$ horizontal exit arcs of the Square gadget of $V_i$ and $V_{d'}$ are paired (one pair per colour class) and directed to a vertex; these $k$ vertices are then each connected by a single arc of capacity $L$ to $t_2$ (cf.\ Figure~\ref{fig:XNLP PW Structure V1V2}).

If the depth of $p$ is even, then we do the same as above, but with vertical entry and exit arcs instead of the horizontal ones.

We now set the demand for commodity~1 to the sum of the capacities of the outgoing arcs of $s_1$ (which is equal to the sum of the capacities of the incoming arcs of $t_1$). We set the demand for commodity~2 to the sum of the capacities of the outgoing arcs of $s_2$ (which is equal to the sum of the capacities of the incoming arcs of $t_2$). This completes the construction. We now prove the bound on the pathwidth, followed by the correctness of the reduction and a discussion of the time and space needed to produce it. To prove the pathwidth bound, we note that all constructed gadgets are pieces, and thus we can apply Lemma~\ref{lem:puzzletw}.

\begin{claim}
The constructed graph has treewidth at most $16k+O(1)$.
\end{claim}
\begin{claimproof}
We construct a tree decomposition as follows. First, we will ensure that $s_1,t_1,s_2,t_2$ are in every bag. Following the proof of Claim~\ref{clm:xnlp-pw-bound}, the pathwidth (and thus the treewidth) of the Triangle gadget and the Square gadget is at most $4k+O(1)$. Recall that the Triangle gadget has $2k$ entry arcs and $2k$ exit arcs, whereas the Square gadget has $4k$ entry arcs and $4k$ exit arcs. The full construction (treating the Triangle and Square gadgets as vertices) has a tree decomposition of width~$1$, since it is tree (see also Figure~\ref{fig:XALP TW Structure full}). Applying Lemma~\ref{lem:puzzletw} and Remark~\ref{rem:puzzle:tw-pw}, we obtain a bound of $16k + O(1)$.
\end{claimproof}
We note that a slightly stronger bound of $4k+O(1)$ seems possible with a more refined analysis, but this bound will be sufficient for our purposes.

\begin{claim}
The given \textsc{Tree Chained Multicolour Clique} instance has a solution if and only if the constructed instance of \ILCF{2} has a solution.
\end{claim}
\begin{claimproof}
If $G$ is a YES-instance of \tcmc, then using a similar construction of flows as in Theorem~\ref{thm:2comXNLP}, we can see that the constructed graph has a flow with all arcs from $s_1$ and $s_2$ and to $t_1$ and $t_2$ are used to capacity. Minor modifications are needed to route the flow corresponding through the tree structure of the Square gadgets in the constructed instance here. Examples of such flow paths are illustrated in Figure~\ref{fig:XALP TW Structure full} for the orange and purple vertices.

Conversely, suppose we have an integer $2$-commodity flow that meets the demands. Hence, there is a $2$-commodity flow in the constructed graph with all arcs from $s_1$ and $s_2$ and to $t_1$ and $t_2$ used to capacity; that is, their total flow is equal to their total capacity, with flow with the corresponding commodity: commodity~1 for $s_1$ and $t_1$, and commodity~2 for $s_2$ and $t_2$. Like in Theorem~\ref{thm:2comXNLP}, all flow of commodity~1 does not leave the Verifier gadget it enters, as the constructed graph again is a acyclic. Then, $s_2$ sends $L$ units of flow of commodity~$2$ to each Vertex selector gadget. This flow now splits into $\alpha$ and $L-\alpha$, where $\alpha \in S(V_{i,l})$, as Lemma~\ref{lemma:verifierlemma} can be applied. Therefore, there is some $v \in V_{i,l}$ such that $\alpha = S(v)$. We select $v$ into the multicoulored clique. Each Vertex selector gadget in turn sends $\alpha$ and $L-\alpha$ flow respectively through the two input arcs of the edge check gadget incident on it. Since there is a flow passing through each Edge check gadget, we know that there exists an edge between the pair of the selected vertices. From the construction, we then see that the selected vertex form a tree chained multicolour clique. Therefore, $G$ is a YES-instance of \tcmc.
\end{claimproof}
As the constructed graph with its capacities can be built using $O(f(k)+\log |V|)$ space for some computable function $f$, the result follows. (See also the discussion at the end of the proof
of Theorem~\ref{thm:2comXNLP}.)
\end{proof}

\subsubsection{Reductions for Undirected Graphs} \label{sec:unary:undirected}
We now reduce from the case of directed graphs to the case of undirected graphs in a general manner, by modification of a transformation by Even et al.~\cite[Theorem~4]{NPCundirected}. In this way, both our hardness results (for parameter pathwidth and for parameter treewidth) can be translated to undirected graphs.

\begin{lemma} \label{lem:unary-undirected-transform}
Let $G$ be a directed graph of an \ILCF{2} instance with capacities given in unary. Then in logarithmic space, we can construct an equivalent instance of \UILCF{2} with an undirected graph $G'$ with $\pw(G') \leq \pw(G) + O(1)$, $\tw(G') \leq \tw(G) + O(1)$, and unit capacities.
\end{lemma}
\begin{proof}
We consider the transformation given by Even et al.~\cite[Theorem~4]{NPCundirected} that shows the NP-completeness of \UILCF{2} and argue that we can modify it to obtain a parameterised transformation from \ILCF{2} to \UILCF{2} with capacities given in unary, and with path- or treewidth as parameter. In particular, the transformation increases the path- or treewidth of a graph by at most a constant.

\begin{figure}
    \centering
    \includegraphics[width=\textwidth]{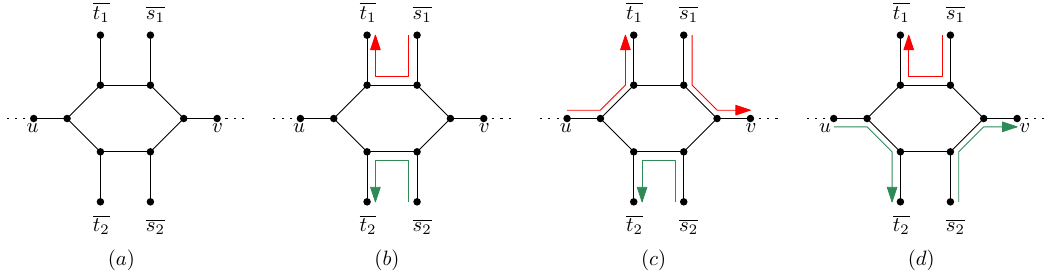}
    \caption{The transformation of Lemma~\ref{lem:unary-undirected-transform} from directed to undirected graphs. Every arc $uv$ with capacity $c$ is replaced by $c$ parallel copies of gadget (a), where $\overline{t_1}, \overline{s_1}, \overline{t_2}, \overline{s_2}$ are the same for every gadget, for all arcs. All capacities are 1. The remaining figures illustrate that the gadget either transports no flow (b), a flow of commodity~1 (c), or a flow of commodity~2 (d).}
    \label{fig:UDE flow types}
\end{figure}

Given a directed graph $G = (V,E)$, demands $d_1$ and $d_2$, and capacity function $c:E \to \mathbb{N}_0$, we construct the instance $G'$, $d'_1$ and $d'_2$, and $c': E(G')\to \{0,1\}$ as follows. To the graph~$G$, we add four new vertices $\overline{s_1}, \overline{s_2}, \overline{t_1}, \overline{t_2}$ as new sources and sinks. We connect $\overline{s_i}$ to $s_i$ and $\overline{t_i}$ to $t_i$ by $d_i$ parallel undirected edges of capacity~$1$, for each $i\in \{1,2\}$. Next, for each arc $uv \in E$ of capacity $p$, we create $p$ parallel undirected edges between $u$ and $v$ of capacity $1$ each. Then, we replace each of these $p$ undirected edges by the following Diamond gadget. We create a cycle $C$ on six vertices $x^{uv}_1,\ldots,x^{uv}_6$, numbered in cyclic order, which we make adjacent to $u$, $\overline{t_1}$, $\overline{s_1}$, $v$, $\overline{s_2}$, and $\overline{t_2}$ respectively by an edge of capacity~$1$ (see Figure~\ref{fig:UDE flow types}(a)). 

This is the graph $G'$ and $c'$ is as just described. In $G'$, the demands on the two commodities are $d'_1 = d_1 + e^*$ and $d'_2= d_2 + e^*$, where $e^*$ is the number of edge gadgets in $G'$ (i.e.~the sum of all capacities in $c$). While $G'$ is technically a multi-graph, any parallel edges $e$ can be subdivided once, and the resulting edges $e_1$ and $e_2$ given the same capacity as $e$. By abuse of notation, we still call this graph $G'$.

\begin{claim}
The pathwidth of $G'$ is $\pw(G) + O(1)$ and the treewidth is $\tw(G) + O(1)$.
\end{claim}
\begin{claimproof}
We note that each of the Diamond gadgets has pathwidth and treewidth~$2$ and forms a path piece and a tree piece. We add $\overline{s_1}, \overline{s_2}, \overline{t_1}, \overline{t_2}$ as well as $s_1$, $s_2$, $t_1$, $t_2$ to every bag. Hence, using Lemma~\ref{lem:puzzle} and~\ref{lem:puzzletw} and the above description, the claim follows.
\end{claimproof}

\begin{claim}
The demands $d_1$ and $d_2$ are met in $G$ if and only if the demands $d'_1$ and $d'_2$ are met in $G'$.
\end{claim}
\begin{claimproof}
Suppose that the demands in the directed graph $G$ are met by some flow. Then, first we send one unit of flow of each commodity in each edge gadget as shown in Figure~\ref{fig:UDE flow types}(b). If $uv$ is used to flow one unit of commodity~1 in $G$, then we change the direction of flow in the edge gadget as in Figure~\ref{fig:UDE flow types}(c). If one unit of commodity~2 flows through $uv$, we change the flow through the edge gadget as in Figure~\ref{fig:UDE flow types}(d). Hence, in addition to the $d_1$ units of flow of commodity~1 and $d_2$ units of flow of commodity~2, $e^*$ units of flow of each type of commodity flows through $G'$. Therefore, the demands $d'_1$ and $d'_2$ are met.

Conversely, suppose that the demands of each commodity are met in the undirected graph $G'$ by some flow. The pattern of flow through each edge gadget could be as in one of the three flows in Figure~\ref{fig:UDE flow types}. If the flow pattern is as in Figure~\ref{fig:UDE flow types}(b), then the corresponding flows through the arc $uv$ in $G$ are set as $f^1(uv)=f^2(uv)=0$. If it is in accordance with Figure~\ref{fig:UDE flow types}(c), then the corresponding flows in $G$ are set as $f^1(uv)=1$ and $f^2(uv)=0$. If the flow pattern is as in Figure~\ref{fig:UDE flow types}(d), then the corresponding flows in $G$ are set as $f^1(uv)=0$ and $f^2(uv)=1$. There are no other options, as every edge incident to any of $\overline{s_1}, \overline{s_2}, \overline{t_1}, \overline{t_2}$ must have $1$ unit of flow of that commodity, otherwise the demands cannot be met. Since the capacity of each edge of the Diamond gadget is $1$, the three options (b), (c), (d) in Figure~\ref{fig:UDE flow types} model exactly the possibilities of sending $1$ unit of flow over each edge incident to one of $\overline{s_1}, \overline{s_2}, \overline{t_1}, \overline{t_2}$. Therefore, the flow through $G$ is at least $d_1 + d_2$ and the demands of each commodity are met.
\end{claimproof}

The construction can be done in logarithmic space: while scanning
$G$, we can output $G'$.
This completes the proof.
\end{proof}

\UntwocomXNLP*
\begin{proof}
The proof of membership in XNLP follows in the same way as the membership of \ILCF{2} as described in the proof of Theorem~\ref{thm:2comXNLP}. For the hardness, apply the reduction of Lemma~\ref{lem:unary-undirected-transform} to the construction of Theorem~\ref{thm:2comXNLP}.
\end{proof}

\UntwocomXALP*
\begin{proof}
The proof of membership in XALP follows in the same way as the membership of \ILCF{2} as described in the proof of Theorem~\ref{thm:2comXALP}. For the hardness, apply the reduction of Lemma~\ref{lem:unary-undirected-transform} to the construction of Theorem~\ref{thm:2comXALP}.
\end{proof}

%---------------------------------------------
% BINARY
%---------------------------------------------

\subsection{Binary Capacities}
We prove our hardness results for {\IMCF} with binary capacities, parameterised by pathwidth. This immediately implies the same results for the parameter treewidth; we do not obtain separate (stronger) results for this case here. Our previous reduction strategy relied heavily on \chgad{a} gadgets, which have size linear in $a$, and thus only work in the case a unary representation of the capacities is given. 

For the case of binary capacities, we can prove stronger results by reducing from {\sc Partition}. However, we need a completely new chain of gadgets and constructions. Therefore, we first introduce a number of new gadgets. After that, we give the hardness results for directed graphs, followed by reductions from the directed case to the undirected case.

As before, throughout the section, all constructions will have disjoint sources and sinks for the different commodities. We will set the demands for each commodity equal to the total capacity of the outgoing arcs from the sources, which is equal to the total capacity of the incoming arcs to the sinks. Thus, the flow over such arcs will be equal to their capacity.

In contrast to the previous, our constructions will have two or three commodities. We name the commodities 1, 2, and 3, with sources $s_1,s_2,s_3$ and sinks $t_1,t_2,t_3$, respectively. We only need the third commodity for the undirected case.

\subsubsection{Gadgets}
We define three different types of (directed) gadgets. Since we use binary capacities, our goal is to double flow in an effective manner. For a given integer $a$, the \emph{$a$-Doubler gadget} receives $a$ flow and sends out $2a$ flow of the same commodity. This gadget is obtained by combining two other gadgets: the $a$-Switch and the Doubling $a$-Switch. The \emph{$a$-Switch gadget} changes the type of flow; that is, it receives $a$ flow from one commodity, but sends out $a$ flow from the other commodity. The \emph{Doubling $a$-Switch} is similar, but sends out $2a$ flow. All three types of gadgets have constant size, even in the binary setting.

We now describe the three gadgets in detail.

\paragraph{$a$-Switch Gadget}
Let $a$ be any positive integer. The first gadget is called an \emph{$a$-Switch}. This gadget turns $a$ units of flow of one type of commodity (in the remainder, of commodity~2) into an equal amount of flow of the other commodity (in the remainder, of commodity~1). 

The gadget is constructed as follows (see Figure~\ref{figure:switch}). We create six vertices $v_1, \dots, v_6$. We add an entry arc incoming to $v_2$ (the left entry arc) and an entry arc incoming to $v_3$ (the right entry arc). We add an exit arc outgoing from $v_4$ (the bottom exit arc) and an exit arc outgoing from $v_5$ (the top exit arc). We add arcs along the paths $v_2v_4v_6t_2$, $x_2v_3v_5v_6$, $s_1v_1v_2$, and $v_1v_3$. All arcs have capacity $a$. We call $v_2, v_3, v_4, v_5$ the boundary vertices of the gadget.

We note that, technically, we could also count the arc incoming on $v_1$ and the arc outgoing from $v_6$ as entry and exit arcs, but since they are coming from $s_1$ and going to $t_2$ respectively, we ignore this aspect.

\begin{figure}
    \centering
    \includegraphics[keepaspectratio,width=0.4\textwidth]{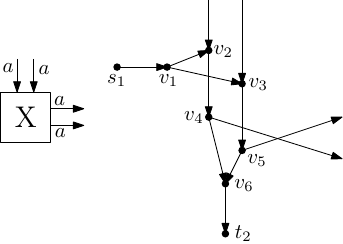}
    \caption{The $a$-Switch gadget. Left: the schematic representation of the gadget, with its entry and exit arcs. Right: the graph that realises the gadget. All arcs have capacity $a$.}
    \label{figure:switch}
\end{figure}

\begin{lemma}\label{lem:aswitch}
Consider the $a$-Switch gadget for some integer $a$. Let $f$ be some $2$-commodity flow such that the arc outgoing from $s_1$ carries $a$ units of flow commodity~1 and the arc incoming to $t_2$ carries $a$ units of flow of commodity~2.
\begin{enumerate}
    \item If the left entry arc carries $a$ units of flow of commodity~2 and the right entry arc carries $0$ units of flow of commodity~2, then the top exit arc carries $a$ units of flow commodity~1 and the bottom exit arc carries $0$ units of flow.
    \item If left entry arc carries $0$ units of flow of commodity~2 and the right entry arc carries $a$ units of flow of commodity~2, then the top exit arc carries $0$ units of flow and the bottom exit arc carries $a$ units of flow of commodity~1.
\end{enumerate}
\end{lemma}
\begin{proof}
Suppose the left entry arc carries $a$ units of flow of commodity~2 and the right entry arc carries $0$ units of flow of commodity~2. We must send $a$ units of commodity~2 over the path $v_2v_4v_6t_2$, as this is the only way $a$ units of commodity~2 can be sent over the arc $v_6t_2$ (recall that this arc must be used to capacity and this flow cannot come from $s_1$). Then all arcs along the path have been used to capacity. By a similar argument, the $a$ units of flow of commodity~1 from $s_1$ to $v_1$ must go to $v_3$, and then via $v_5$ through the top exit arc.

The other case is symmetric.
\end{proof}

\begin{lemma}\label{lem:aswitchpw}
For any integer $a$, the $a$-Switch is a path piece such that the required path decomposition (ignoring the sources and sinks) has width~$2$.
\end{lemma}
\begin{proof}
The gadget is a piece by construction, with $B^- = \{v_2, v_3\}$ and $B^+ = \{v_4,v_5\}$. To construct the path decomposition, start with a bag containing $v_2, v_3$. We can now add $v_1$, and in the next bag remove $v_1$ and add $v_4$. Then remove $v_2$ and add $v_5$. In the final bag, remove $v_3$ and add $v_6$. Each bag contains at most 3 vertices.
\end{proof}

\paragraph{Doubling $a$-Switch Gadget}
Let $a$ be any positive integer. The second gadget is called a \emph{Doubling $a$-Switch}. This gadget turns $a$ units of flow of one type of commodity (in the remainder, of commodity~1) into a $2a$ units of flow of the other commodity (in the remainder, of commodity~2).

The gadget is constructed as follows (see Figure~\ref{figure:daswitch}). We create fourteen vertices $v_1, \dots, v_{14}$. We add an entry arc incoming to $v_4$ (the left entry arc) and an entry arc incoming to $v_9$ (the right entry arc), each of capacity~$a$. We add an exit arc outgoing from $v_{13}$ (the bottom exit arc) and an exit arc outgoing from $v_{14}$ (the top exit arc), each of capacity $2a$.
We add arcs with capacity $a$ along the paths $v_4v_5v_6v_7v_8t_1$, $v_9v_{10}v_{11}v_{12}v_8$. We also add arcs $v_2v_4$, $v_2v_6$, $v_3v_9$, $v_3v_{11}$, $v_5v_{13}$, $v_7v_{13}$, $v_{10}v_{14}$ and $v_{12}v_{14}$ with capacity $a$. Finally, we add the arcs $s_2v_1$, $v_1v_2$, $v_1v_3$ with capacity $2a$. The vertices $v_4$, $v_9$, $v_{13}$, $v_{14}$ are the boundary vertices of the gadget.

\begin{figure}
    \centering
    \includegraphics{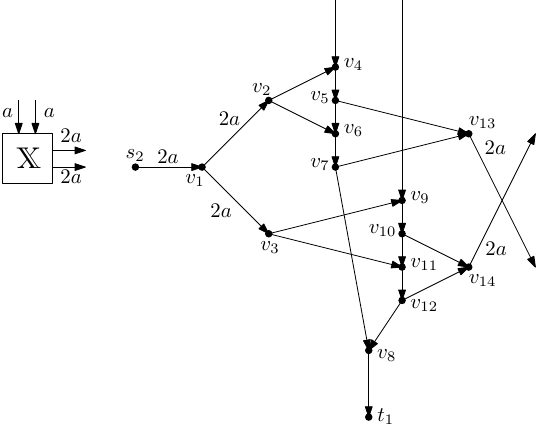}
    \caption{The Doubling $a$-Switch gadget. Left: the schematic representation of the gadget, with its entry and exit arcs. Right: the graph that realises the gadget. The arcs are labelled with their capacities; all unlabelled arcs have capacity $a$.}
    \label{figure:daswitch}
\end{figure}

\begin{lemma}\label{lem:doublingswitch}
Consider the Doubling $a$-Switch gadget for some integer $a$. Let $f$ be some $2$-commodity flow such that the arc outgoing from $s_2$ carries $2a$ units of flow of commodity~1 and the arc incoming to $t_1$ carries $a$ units of flow of commodity~2 .
\begin{enumerate}
    \item If the left entry arc carries $a$ units of flow of commodity~1 and the right entry arc from carries $0$ units of flow of commodity~1, then the top exit arc carries $2a$ units of flow of commodity~2, and the bottom exit arc carries $0$ units of flow of commodity~2.
    \item If the left entry arc carries $0$ units of flow of commodity~1 and the right entry arc carries $a$ units of flow of commodity~1, then the top exit arc carries $0$ units of flow of commodity~2 and the bottom exit arc carries $2a$ units of flow of commodity~2.
\end{enumerate}
\end{lemma}
\begin{proof}
    Suppose the left entry arc carries $a$ units of flow of commodity~1 and the right entry arc carries $0$ units of flow of commodity~1. We must send $a$ units of flow of commodity~1 over the path $v_4v_5v_6v_7v_8t_1$, as this is the only way $a$ units of commodity~1 can be sent over the arc $v_8t_1$. Then all arcs along the path have been used to capacity. This implies that the flow from $s_2$ to $v_1$ must go to $v_3$, and then via $v_9v_{10}$ and $v_{11}v_{12}$ to $v_{14}$ after which it must go through the top exit arc.

    The other case is symmetric.
\end{proof}

\begin{lemma}\label{lem:doublingswitchpw}
For any integer $a$, the Doubling $a$-Switch is a path piece such that the required path decomposition (ignoring sources and sinks) has width~$5$.
\end{lemma}
\begin{proof}
The gadget is a piece by construction, with $B^- = \{v_4,v_{9}\}$ and $B^+ = \{v_{13},v_{14}\}$. To construct the path decomposition, start with a bag containing $v_4, v_9, v_1, v_2, v_3$, followed by a bag containing $v_4,v_9,v_2,v_3$. From there, we create bags where we subsequently add $v_5$, remove $v_4$, add $v_6$ and $v_{13}$, remove $v_5$, add $v_7$, remove $v_2$ and $v_6$, and add $v_8$. The bag then contains $v_3,v_8,v_9,v_{13}$. Then we create bags where we subsequently add $v_{10}$, remove $v_9$, add $v_{11}$ and $v_{14}$, remove $v_{10}$, and add $v_{12}$. This forms the required path decomposition. All bags contain at most six vertices.
\end{proof}

\paragraph{$a$-Doubler Gadget}
Let $a$ be any positive integer. We can combine an $a$-Switch gadget with a Doubling $a$-Switch gadget to get an {\em $a$-Doubler} gadget. The first gadget changes the commodity of the flow, where the second gadget changes the commodity back with double the amount of flow. 

We construct this gadget as follows (see Figure~\ref{figure:doubler}). Create an $a$-Switch gadget and an Doubling $a$-Switch gadget (refer to Figure~\ref{figure:switch} and~\ref{figure:daswitch}). Identify the top exit arc of the $a$-Switch gadget with the left entry arc of the Doubling $a$-Switch gadget. Then identify the bottom exit arc of the $a$-Switch gadget with the right entry arc of the Doubling $a$-Switch gadget.

\begin{figure}
    \centering
    \includegraphics{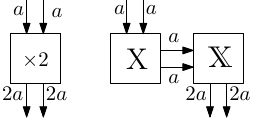}
    \caption{The $a$-Doubler gadget. Left: the schematic representation of the gadget, with its entry and exit arcs. Right: the graph that realises the gadget. The arcs are labelled by their capacities.}
    \label{figure:doubler}
\end{figure}

Note that the $a$-Doubler gadget has two entry arcs (the left and right entry arcs) and two exit arcs (the left and right exit arcs), corresponding to the left and right entry arcs of the $a$-Switch gadget and the top and bottom exit arcs of the Doubling $a$-Switch gadget respectively.

\begin{lemma}\label{lem:adoubler}
Consider the $a$-Doubler gadget for some integer $a$. Let $f$ be some $2$-commodity flow. Then:
\begin{enumerate}
\item If the left entry arc carries $a$ units of flow of commodity~2 and the right entry arc carries $0$ units of flow of commodity~2, then the left exit arc carries $2a$ units of flow of commodity~2 and the right exit arc carries $0$ units of flow of commodity~2.
\item If the left entry arc carries $0$ units of flow  of commodity~2 and the right entry arc carries $a$ units of flow of commodity~2, then the left exit arc carries $0$ units of flow of commodity~2 and the right exit arc carries $2a$ units of flow of commodity~2.
\end{enumerate}
\end{lemma}
\begin{proof}
The lemma follows immediately by combining Lemma~\ref{lem:aswitch} and~\ref{lem:doublingswitch}.
\end{proof}

\begin{lemma}\label{lem:adoublerpw}
For any integer $a$, the $a$-Doubler is a path piece such that the required path decomposition (ignoring sources and sinks) has width~$5$.
\end{lemma}
\begin{proof}
Recall from Lemma~\ref{lem:aswitchpw} that the $a$-Switch is a path piece of pathwidth at most~$2$ with two entry arcs and two exit arcs. Recall from Lemma~\ref{lem:doublingswitchpw} that the Doubling $a$-Switch is a path piece of pathwidth at most~$5$ with two entry arcs and two exit arcs. Note that the structure of the $a$-Doubler gadget trivially satisfies the preconditions of Lemma~\ref{lem:puzzle2}. Hence, the lemma then follows by applying Lemma~\ref{lem:puzzle2}.
\end{proof}

\subsubsection{Reduction for Directed Graphs}
With the gadgets in hand, we can prove our hardness result for {\IMCF} (i.e.~the case of directed graphs) for parameter pathwidth.

\twocomNPC*
\begin{proof}
Membership in NP is trivial. To show NP-hardness, we transform from {\sc Partition}. Recall {\sc Partition} problem asks, given positive integers $a_1,\ldots,a_n$, to decide if there is a subset $S\subseteq [n]$ with $\sum_{i\in S} a_i = B$, where $B= \sum_{i=1}^n a_i/2$. This problem is well known to be NP-complete~\cite{Karp72}.

So consider an instance of {\sc Partition} with given integers $a_1,\ldots,a_n$. Create the sources $s_1, s_2$ and the sinks $t_1,t_2$. Create two vertices $b_1, b_2$, both with an arc of capacity $B$ to $t_2$.

For each $a_i$, we build a Binary gadget that either sends $a_i$ units of flow to $b_1$ or $a_i$ units of flow to a vertex $b_2$, in each case of commodity~2. This will indicate whether or not $a_i$ is in the solution set to the {\sc Partition} instance. This gadget is constructed as follows (see Figure~\ref{figure:binarygadget} for the case when $a_i=13$). Consider the binary representation $a_i^p,\ldots,a_i^0$ of $a_i$. That is, $a_i = \sum_{j=0}^p 2^j a_i^j$. For each $j \in [p]$ such that $a_i^j = 1$, we create a column of chained Doubler gadgets. For each $j' < j$, create a $2^{j'}$-Doubler gadget and identify its entry arcs with the exit arcs of the $2^{j'-1}$-Doubler gadget (see Figure~\ref{figure:binarygadget}). Then the left exit arc of the (final) $2^{j-1}$-Doubler gadget is directed to $b_1$, while the right exit arc is directed to $b_2$. 

Note that the Binary gadget for $a_i$ still has $2 \sum_{j=0}^p a_i^j$ open entry arcs (of the $1$-Doubler gadget of each column). These naturally partition into $\sum_{j=0}^p a_i^j$ left entry arcs and $\sum_{j=0}^p a_i^j$ right entry arcs. These all have capacity~$1$. We now connect these arcs. All further arcs in the construction will have capacity~$1$.

Create two directed paths $P_i^1,P_i^2$ of $2 \sum_{j=0}^p a_i^j$ vertices each (see Figure~\ref{figure:binarygadget}). We consider the vertices of each of these paths in consecutive pairs, one pair for each $a_i^j$ that is equal to~$1$. For each $j \in [p]$ such that $a_i^j = 1$, create a vertex $v_i^j$ with an arc from $s_2$, an arc to the first vertex of the pair on $P_i^1$ corresponding to $a_i^j$, and an arc to the first vertex of the pair on $P_i^2$ corresponding to $a_i^j$. Then, add an arc from the second vertex of the pair on $P_i^1$ corresponding to $a_i^j$ to the left entry arc of the $1$-Doubler gadget of the $j$th column of gadgets and an arc from the second vertex of the pair on $P_i^2$ corresponding to $a_i^j$ to the right entry arc of the $1$-Doubler gadget of the $j$th column of gadgets. Finally, create a vertex $u_i$ with an arc to the first vertex of $P_i^1$ and to the first vertex of $P_i^2$ and create a vertex $w_i$ with an arc from the last vertex of $P_i^1$ and the last vertex of $P_i^2$. This completes the description of the Binary gadget.

We now chain the Binary gadgets. For each $i \in [n-1]$, add an arc from $w_i$ to $u_{i+1}$. Add an arc from $s_1$ to $u_1$ and from $w_n$ to $t_1$. These arcs all have capacity~$1$. 

We now set the demand for commodity~1 to the sum of the capacities of the outgoing arcs of $s_1$ (which is equal to the sum of the capacities of the incoming arcs of $t_1$). We set the demand for commodity~2 to the sum of the capacities of the outgoing arcs of $s_2$ (which is equal to the sum of the capacities of the incoming arcs of $t_2$). This completes the construction.

\begin{figure}[tb]
    \centering
    \includegraphics{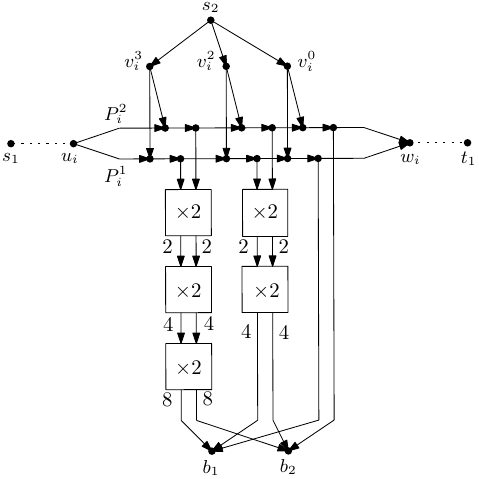}
    \caption{Example of the Binary gadget for $a_i = 13$ and its associated paths $P_i^1$ and $P_i^2$ and the ends $u_i$ and $w_i$. Since $13 = 2^3 + 2^2 + 2^0$, we have a column with three Doubler gadgets, a column with two Doubler gadgets, and one with no Doubler gadgets. The vertices $v_i^3$, $v_i^2$ and $v_i^0$ are also drawn. Arcs are labelled by their capacities, but unlabelled arcs have capacity $1$. If $1$ unit of flow of commodity~1 is sent from $s_1$ to $t_1$, then it must pick one of $P_i^1, P_i^2$ to go through. Hence, the gadget ensures that either $a_i$ units of flow of commodity~2 are sent to $b_1$ through the left entry and exit arcs of the Doubler gadgets, or $a_i$ units of flow of commodity~2 are sent to $b_2$ through the left entry and exit arcs of the Doubler gadgets.}
    \label{figure:binarygadget}
\end{figure}

\begin{claim}
The constructed graph has pathwidth at most~$13$.
\end{claim}
\begin{claimproof}
We construct a path decomposition as follows. Add $s_1, s_2, t_1, t_2, b_1$ and $b_2$ to every bag. Then, we construct a path decomposition for the Binary gadget and its associated paths $P_i^1$ and $P_i^2$. We create the trivial path decompositions for $P_i^1$ and $P_i^2$ and union each of their bags, so that we `move' through the two paths simultaneously. When the first vertex of the pair corresponding to $a_i^j$ (where $a_i^j = 1$) is introduced in a bag, we add a subsequent copy of the bag to which we add $v_i^j$ and another subsequent copy without it. Then, when the second vertex of the pair corresponding to $a_i^j$ (where $a_i^j = 1$) is introduced in a bag, we add bags for the Doubler gadgets of the column corresponding to $a_i^j$. Since each Doubler gadget has two entry and exit arcs and is a path piece with a path decomposition of width~$5$ by Lemma~\ref{lem:adoublerpw}, it follows from Lemma~\ref{lem:puzzle2} (recalling Remark~\ref{rem:puzzle}) that each column has pathwidth~$5$. Combining this with the other vertices we add to each bag ($s_1, s_2, t_1, t_2, b_1$ and $b_2$) and to each bag for each column (both second vertices of the pair corresponding to the column), the total width of the path decomposition is~$13$.
\end{claimproof}

\begin{claim}
The given \textsc{Partition} instance has a solution if and only if the constructed instance of \ILCF{2} has a solution.
\end{claim}
\begin{claimproof}
Let $S \subseteq [n]$ be a solution to the {\sc Partition} instance. We will find a corresponding solution to the constructed \ILCF{2} instance. For each $i \in [n]$, we do the following. If $i \in S$, then we send flow of commodity~2 from $s_2$ to $b_1$, through left entry and exit arcs of the Doubler gadgets in the Binary gadget corresponding to $a_i$. To reach this left side of the Doubler gadgets, the flow passes through vertices and arcs of $P_i^1$. We can thus send flow of commodity~1 from $u_i$ to $w_i$ via $P_i^2$. Otherwise, if $i \not\in S$, we send flow of commodity~2 from $s_2$ to $b_2$, through right entry and exit arcs of the Doubler gadgets in the Binary gadget corresponding to $a_i$. To reach this right side of the Doubler gadgets, the flow passes through vertices and arcs of $P_i^2$. We can thus send flow of commodity~1 from $u_i$ to $w_i$ via $P_i^1$. 

Now note that by the properties of the Doubler gadget, proved in Lemma~\ref{lem:adoubler}, $b_1$ will receive $a_i$ units of flow of commodity~2 if $i \in S$ and $b_2$ will receive $a_i$ units of flow of commodity~2 if $i \not\in S$. Since $S$ is a solution to {\sc Partition}, both $b_1$ and $b_2$ receive $B$ units of flow of commodity~2, which they can then pass on to $t_2$. Moreover, we observe that we can send $1$ unit of flow from $s_1$ to $t_1$ via the paths $P_i^1$ and $P_i^2$, using $i \notin S$ and $i \in S$ respectively.

In the other direction, suppose we have an integer $2$-commodity flow that meets the demands. That is, there is a $2$-commodity flow in the constructed graph with all arcs from $s_1$ and $s_2$ and to $t_1$ and $t_2$ used to capacity; that is, their total flow is equal to their total capacity, with flow with the corresponding commodity: commodity~1 for $s_1$ and $t_1$, and commodity~2 for $s_2$ and $t_2$. Hence, the arc $w_nt_1$ is used to capacity, so $1$ unit flow of commodity~1 flows over this arc. Because of the direction of the arcs, this flow can only come from $u_n$, and so from $w_{n-1}$. By induction, this flow must come over the arc $s_1u_1$. We see that the flow of commodity~1 starting at $u_1$ takes a path which is a union of $P_{i}^{j_i}$ paths, for $i\in [n]$ and $j_i\in \{1,2\}$. In particular, this flow does not `leak' into any Doubler gadget, uses all the arcs $w_iu_{i+1}$ for all $i\in [n-1]$ to capacity, and uses a complete path $P_i^{j_i}$ up to capacity for each $i\in[n]$, $j_i\in \{1,2\}$.
Consider the Binary gadget corresponding to $a_i$. By the above argument, we must have an $1$ unit of flow of commodity~1 going through one of the two paths $P_i^1$ or $P_i^2$, also using the arc $w_iu_{i+1}$ to capacity. Suppose this is $P_i^1$. This means that any flow from $s_2$ to $t_2$, in this Binary gadget, has to utilise $P_i^2$, the right side of the Doubler gadgets, and end up at $b_2$. We can apply Lemma~\ref{lem:adoubler} to every Doubler gadget. As flow of commodity~2 is carried over the right entry and exit arcs and no flow flows over the left entry and arcs, the total flow value reaching $b_2$ has to be equal to $a_i$. The same argument holds with respect to $P_i^2$ and $b_1$. Let $S \subseteq [n]$ be the set of indices $i$ for which the flow of commodity~2 through the Binary gadget corresponding to $a_i$ arrives at $b_1$. Since the edge $b_1t_2$ has capacity $B$ and since $b_1$ has received $\sum_{i \in S} a_i$ units of flow of commodity~2, we find that $\sum_{i \in S} a_i \leq B$. Similarly, $\sum_{i \notin S} a_i \leq B$. Since $\sum_{i \in [n]} a_i = 2B$, we conclude that $S$ is a valid solution to the {\sc Partition} instance.
\end{claimproof}

Finally, as each $a$-Doubler has constant size, the gadget for some $a_i$ has size $O(\log^2(a_i))$, which is polynomial in the input size. Hence, the construction as a whole has size polynomial in the input size. Moreover, it can clearly be computed in polynomial time.
\end{proof}

\subsubsection{Reduction for Undirected Graphs}
We now reduce from the case of directed graphs to the case of undirected graphs in a general manner. We define a new gadget that is similar to the gadget we used in Section~\ref{sec:unary:undirected}. However, we note that there we required $a$ copies of the gadget if the capacity of an arc is $a$, which is not feasible in the case of binary capacities. Also note that increasing the capacities of the gadget by Even et al.~\cite[Theorem~4]{NPCundirected}, here Figure~\ref{fig:UDE flow types}, invalidates the gadget, as any under-capacity edge would allow flow in the other direction. Hence, we need a different gadget, before we can give our hardness result for undirected graphs.

We first define a subgadget called a \emph{Directed twin flow edge} gadget, seen in Figure~\ref{fig:DirTwnFlw}, which we construct as follows. Given two vertices $u$ and $v$, two commodities $i_1$ and $i_2$, and an integer (capacity) $c$, we add vertices $w_1, w_2, w_3$, and $w_4$. We then add edges $uw_1$, $w_1w_2$, $w_2w_3$, $w_3w_4$, $w_4v$, $t_{i_1}w_1$, $s_{i_1}w_2$, $t_{i_2}w_3$ and $s_{i_2}w_4$, all with capacity $2c$.
We will refer to the vertex $u$ as the \emph{tail-vertex} of the gadget and to $v$ as the \emph{head-vertex}. Furthermore, we say that the gadget is \emph{labelled} by the two commodities whose source and sink are attached in the gadget.

\begin{figure}
    \centering
    \includegraphics[width = \textwidth]{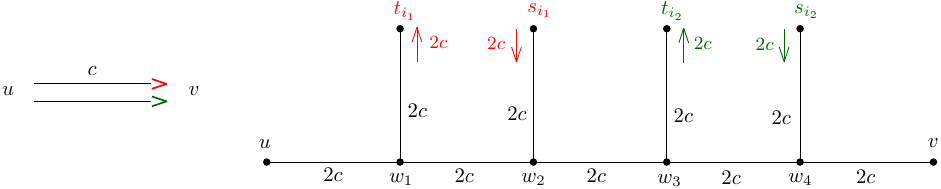}
    \caption{The Directed twin flow edge gadget of capacity $c$. Left: the schematic representation of the gadget. Right: the gadget itself. The arcs are labelled by their capacities. The colours are used to indicate the respective commodities.}
    \label{fig:DirTwnFlw}
\end{figure}

The gadget only allows flow of the two given commodities through, in only one direction and only in equal amounts. The following lemma captures this behaviour.

\begin{lemma} \label{lem:DirTwnEdg}
Let $G$ be an undirected graph is part of an instance of \UILCF{\ell} for some $\ell \geq 2$. Let $u,v \in V(G)$ and suppose there is a a Directed twin flow edge gadget $H$ in $G$ for commodities $i_1$ and $i_2$ with $u$ as the tail-vertex of $H$ and $v$ the head-vertex of $H$. Let $f$ be an $\ell$-flow in $G$ that fills all edges incident on each source and sink to capacity with flow of the corresponding commodity. Then:
\begin{itemize}
    \item no flow of any commodity other than $i_1$ or $i_2$ can travel from $u$ to $v$, through $H$,
    \item the amount of flow of commodities $i_1$ and $i_2$ travelling from $u$ to $v$ through $H$ is equal,
    \item no flow can travel through $H$ from $v$ to $u$.
\end{itemize}
\end{lemma}

\begin{proof}
    For simplicity, we will assume that the flow $f$ does not go back on itself, i.e.\ no edge has nonzero flow of the same commodity in both directions.

    If in total $x$ units of flow of any commodity other than $i_1$ or $i_2$ travel from $u$ to $w_3$ through $H$, then by the capacity of $w_1w_2$ and the fact that $2c$ units of flow of commodity~$i_1$ enter the gadget at $w_2$, at least $x$ units of flow of commodity~$i_1$ must travel from $w_2$ via $w_3$ to $w_4$. Hence, by a similar argument, at least $2x$ units of flow of commodity~$i_2$ leave the gadget at $w_4$. Therefore, the flow of commodity $i_2$ travelling from $u$ to $w_3$ must be at least $2x$ units, a contradiction.

    We then note that if we have $a_1$ units of flow of commodity $i_1$ entering the gadget at $w_1$, we can only have at most a total of $a_1$ units of the other commodities flowing from $w_1$ to $w_2$. Indeed, note that $2c$ units of flow of commodity~$i_1$ must enter the gadget from $s_{i_1}$ to $w_2$ and leave the gadget via $w_1$ to $t_{i_1}$. Hence, $2c-a_1$ units of flow of commodity~$i_1$ would travel from $w_2$ to $w_1$. Similarly, if we have $a_2$ units of flow of commodity $i_2$ leaving the gadget at $w_4$, we can only have at most a total of $a_2$ units of the other commodities flowing from $w_3$ to $w_4$. Hence, the amount of flow of commodities $i_1$ and $i_2$ travelling from $u$ to $v$ through $H$ is equal.

    We also find that the edges $w_1w_2$ and $w_3w_4$ are always used to capacity and thus we cannot send any flow from $v$ to $u$ through $H$, which proves the last item of the lemma.
\end{proof}

We now create a \emph{Directed edge gadget} that effectively functions as a directed edge. To enable this gadget, we need an additional commodity to activate the gadget. So suppose we have $\ell+1$ commodities, with commodity $\ell+1$ being this extra commodity. The gadget is constructed as follows (see Figure~\ref{fig:DirEdgGad}). Given two vertices $u$ and $v$ and an integer (capacity) $c$, we wish to simulate the arc $e=uv$. Create a new vertex $v_e$. For each $i \in [\ell]$, we then attach a Directed twin flow edge gadget from $u$ to $v_e$ for commodities $i$ and $\ell+1$, with capacity $c$. Finally, we add an edge $v_ev$ with capacity $2c$. Again, we will refer to the vertex $u$ as the \emph{tail-vertex} of the gadget and to $v$ as the \emph{head-vertex} of the gadget.

\begin{figure}
    \centering
    \includegraphics[keepaspectratio,width=.4\textwidth]{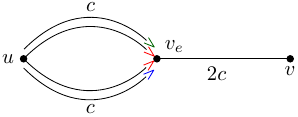}
    \caption{Directed edge gadget, for $\ell = 2$ and capacity $c$.}
    \label{fig:DirEdgGad}
\end{figure}

\begin{lemma}\label{lem:DirEdg}
Let $G$ be an undirected graph that is part of an instance of \UILCF{\ell+1} for some $\ell \geq 2$. Let $u,v \in V(G)$ and suppose there is a a Directed edge gadget $H$ in $G$ with $u$ as the tail-vertex of $H$ and $v$ the head-vertex of $H$. Let $f$ be an $(\ell+1)$-flow in $G$. Then:
\begin{itemize}
    \item in total, no more than $c$ units of flow of commodities $1, \dots, \ell$ can travel from $u$ to $v$, through $H$,
    \item the amount of flow of commodity $\ell + 1$ travelling from $u$ to $v$ through $H$ is equal to the sum of all other commodities travelling from $u$ to $v$ through $H$,
    \item no flow can travel through $H$ from $v$ to $u$.
\end{itemize}
\end{lemma}

\begin{proof}
    Recall that $H_1, \dots, H_\ell$ are the Directed twin flow edge gadgets in the Directed edge gadget and let $v_e$ be the `middle vertex' of the gadget. By the last item of Lemma~\ref{lem:DirTwnEdg}, no flow can travel from $v_e$ to $u$ through any of the $H_i$ and thus no flow can travel from $v$ to $u$ through $H$.

    By the first two items of Lemma~\ref{lem:DirTwnEdg}, the amount of flow travelling from $u$ to $v_e$ through $H_i$ is equal to some $a_i$ for commodities $i$ and $\ell + 1$, and $0$ for any other commodities. We find that the amount of flow travelling from $u$ to $v_e$ of commodity $\ell + 1$ is equal to the sum the other commodities and thus the amount of flow of commodity $\ell + 1$ travelling from $u$ to $v$ through $H$ is equal to the sum of all other commodities travelling from $u$ to $v$ through $H$.

    The first item now holds trivially, since the edge $v_ev$ has capacity $2c$.
\end{proof}

This construction now allows us to extend any reductions from some problem to \ILCF{\ell} to a reduction from the problem to \UILCF{\ell+1}. We now show such an extension only increases the pathwidth by a constant.

\begin{lemma}\label{lem:diredgpw}
Let $G$ be a directed graph of an \ILCF{\ell} instance with a path decomposition of width~$w$, such that each bag contains the sources and sinks of commodities $1,\ldots,\ell$. Then in polynomial time, we can construct an equivalent instance of \UILCF{\ell+1} of pathwidth at most $w+5$.
\end{lemma}
\begin{proof}
For each $i \in [\ell]$, let $d_i$ denote the demand for commodity~$i$ for the given instance of \ILCF{\ell}. To start, we create a source $s_{\ell+1}$ and sink $t_{\ell+1}$ for the extra commodity, with demand $d'_{\ell+1} = \sum_{i=1}^\ell d_i$. Then, for $i \in [\ell]$, we connect $s_{\ell+1}$ to $s_i$ by an edge with capacity $d_i$, and connect $t_i$ to $t_{\ell+1}$ by an edge with capacity $d_i$. Then, replace each arc $uv$ in $G$ by a Directed edge gadget. Call the resulting graph $G'$. Set the remaining demands $d_i' = d_i$ for $i \in \ell$. This completes the construction.

We note that $G'$ can be constructed in polynomial time from $G$. The fact that the given instance of \ILCF{\ell} and the constructed instance of \UILCF{\ell+1} are equivalent follows immediately from Lemma~\ref{lem:DirEdg} and the construction of $G'$. Indeed, by the setting of the demand $d'_{\ell+1}$ and the capacities of the edges incident on $s_{\ell+1}$, every source $s_i$ of $G$ receives $d_i$ units of flow of commodity~$\ell+1$. Hence, every flow of commodity~$i$ is and can be accompanied by an equal amount of flow of commodity~$\ell+1$. Then, following Lemma~\ref{lem:DirEdg}, the direction of $uv$ is maintained by the transformation.

To prove the upper bound on the pathwidth of $G'$, consider some arc $e=uv$ in the directed graph $G$. In the path decomposition of $G$, there must be some bag $X$ that contains both $u$ and $v$. Create $3\ell$ copies of the bag $X$ that we insert after $X$. To each of these copies, add the vertex $v_e$ of the Directed edge gadget (this covers in particular the edge $v_ev$ of the gadget). For each commodity $i \in [\ell]$, add to consecutive bags the pairs $(w_1,w_2)$, $(w_2,w_3)$, and $(w_3,w_4)$ of the Directed twin flow edge gadget $H_i$ in the Directed edge gadget. These three bags for each $i \in [\ell]$ handle the path decomposition for the Directed twin flow edge of commodity $i$. So, in total, the $3\ell$ bags handle the path decomposition for the entire gadget.

We do this for every arc in $G$, where the copies we make of bags are copies only of bags in the path decomposition of $G$. After we have done this for every arc, we add the source and sink of the `extra' commodity $\ell+1$ to every bag. We see that the maximum number of vertices in any bag increases by at most~$5$.
\end{proof}

By combining Lemma~\ref{lem:diredgpw} and Theorem~\ref{theorem:binarypathwidth}, we obtain the following.

\UnthreecomNPC*
\begin{proof}
By Theorem~\ref{theorem:binarypathwidth}, \ILCF{2} with capacities given in binary is NP-complete for graphs of pathwidth at most~$13$. By inspection of the proof, we see that it has a path decomposition of width~$13$ such that each bag contains both sources and both sinks. Then the reduction of Lemma~\ref{lem:diredgpw} immediately implies the theorem.
\end{proof}

%-----------------------------------------------
% VERTEX COVER
%-----------------------------------------------

\subsection{Parameterisation by Vertex Cover --- At Most One Commodity Per Edge}
Our last hardness result concerns a parameterisation by vertex cover. In this case, we will add the additional constraint that the edges must be \emph{monochrome}, that is, for each edge, only one commodity can have positive flow. Later, in Section~\ref{subsec:AlgVC}, we will see that it is possible to approximately solve the problem without this constraint, in polynomial time.

\begin{figure}
    \centering
    \includegraphics{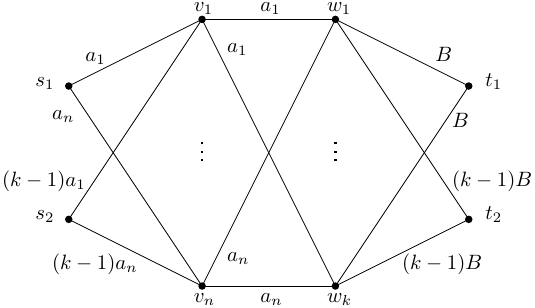}
    \caption{Construction for the reduction from {\sc Bin Packing} to {\sc 2-Commodity Flow}.}
    \label{fig:BinToMono}
\end{figure}

\UnILCFME*

\begin{proof}
    We reduce from {\sc Bin Packing}. Recall that in {\sc Bin Packing}, we are given integers $a_1, \ldots, a_n$, an integer $B$, and an integer $k$. We are asked to decide whether the integers $a_1, \ldots, a_n$ can be partitioned into at most $k$ bins, such that the sum of the numbers assigned to each bin does not exceed $B$. {\sc Bin Packing} is NP-hard for binary weights and $k=2$~\cite{GarJohn1990} and W[1]-hard for unary weights when parameterised by $k$ \cite{JansenKMS13}. In both cases, we may assume that $\sum_{j=1}^n a_j = kB$: if $k=2$, we have the \textsc{Subset Sum} problem, and in the unary case, we can add $kB-\sum_{j=1}^n a_j$ additional integers that are all equal to $1$.
    
    We first describe our construction. Suppose we are given an instance $((a_j)_{1}^n, B, k)$ of {\sc Bin Packing}. Without loss of generality, we may assume that $\sum_{j=1}^n a_j = kB$. We construct an equivalent instance of \UILCFME{2} that has a vertex cover of size $k + 4$ (see also Figure \ref{fig:BinToMono}). We create a graph $G$ with the following vertices and edges:   
    \begin{align*}
        V(G)    &= \{s_1, s_2, t_1, t_2, v_1, \dots v_n, w_1, \dots, w_k\} \\
        E(G)    &= \{s_iv_j : i \in \{1, 2\}, j \in [n]\}
                \cup \{w_iv_j : i \in [k], j \in [n]\}
                \cup \{w_it_j : i \in [k], j \in \{1, 2\}\}.
    \end{align*}
    We set the capacities of the edges as follows:
    \begin{align*}
        c(s_1v_j) &= a_j &\forall j\in [n] \\
        c(s_2v_j) &= (k-1)a_j &\forall j\in [n]\\
        c(w_iv_j) &= a_j &\forall i\in [k],j\in [n]\\
        c(w_it_1) &= B &\forall i\in [k]\\
        c(w_it_2) &= (k-1)B &\forall i\in [k].
    \end{align*}
    We set the demands as $d_1 = kB$ and $d_2 = (k-1)kB$. This completes the reduction.

    Suppose that the instance of {\sc Bin Packing} is a YES-instance. Thus, there is a partition $I_1,\ldots,I_k$ of $\{1,\ldots,n\}$ such that $\sum_{j \in I_i} a_j \leq B$ for each $i \in [k]$. As $\sum_{j=1}^n a_j = kB$, actually $\sum_{j \in I_i} a_j = B$ for each $i \in [k]$. For each $i \in [k]$ and each $j \in I_i$, route $a_j$ units of flow of commodity~1 from $s_1$ to $v_j$ to $w_i$ to $t_1$. Hence, edges incident on $s_1$ and $t_1$ only transport flow of commodity~1. For $i \in [k]$, since $\sum_{j \in I_i} a_j = B$, the flow over each edge $w_it_1$ does not exceed $B$. Also, route $(k-1)a_j$ units of flow of commodity~2 from $s_2$ to $v_j$. Then, split this flow into $k-1$ units of flow to each vertex in $\{w_1,\ldots,w_k\} \setminus \{w_i\}$ and from these vertices to $t_2$. Hence, edges incident on $s_2$ and $t_2$ only transport flow of commodity~2. For $i \in [k]$, $\sum_{i \in [k] \setminus \{i\}} \sum_{j \in I_{i'}} \leq (k-1)B$ since $\sum_{j \in I_i} a_j = B$, and thus the flow over each each edge $w_it_2$ does not exceed $(k-1)B$. Finally, note that edges between $\{v_1,\ldots,v_n\}$ and $\{w_1,\ldots,w_k\}$ only transport flow of a single commodity, because $I_1,\ldots,I_k$ is a partition of $\{1,\ldots,n\}$. Hence, the instance of \UILCFME{2} is a YES-instance.

    Conversely, suppose that the instance of \UILCFME{2} is a YES-instance and consider a valid flow. We analyse the structure of this flow. The cut $(\{s_1\}, V(G) \setminus \{s_1\})$ has capacity $\sum_{i=1}^n a_i = kB = d_1$. Hence, for each $j \in [n]$, we must have $a_j$ units of flow of commodity~1 from $s_1$ to $v_j$. Similarly, the cut $(\{s_2\}, V(G) \setminus \{s_2\})$ has capacity $\sum_{i=1}^n (k-1) a_i = kB = d_2$. Hence, for each $j \in [n]$, we must have $(k-1)a_j$ units of flow of commodity~2 from $s_2$ to $v_j$. Since the total capacity of the remaining edges connected to $v_j$ (the edges $v_jw_i$) is equal to $ka_j$ and since each edge can only be used by one commodity, we find that exactly one of these edges transports $a_j$ units of flow of commodity~1 to some $w_i$ and the other $k-1$ edges each transport $a_j$ units of flow of commodity~2 to the other $w_{i'}$ with $i' \not= i$. Finally, any unit of commodity~1 in $w_i$ must be sent to $t_1$ and any unit of commodity~2 must be set to $t_2$, since by construction the total capacity of the cut $(V(G) \setminus \{t_1, t_2\}, \{t_1, t_2\})$ is $k\sum_{j=1}^n a_j = d_1 + d_2$, so we cannot have any flow going in the reverse direction.
    
    We now construct sets $I_1,\ldots,I_k$. Now observe that each $v_j$ sends $a_j$ units of flow of commodity~1 to exactly one neighbour $w_i$. Then add $j$ to $I_i$. Then $I_1,\ldots,I_k$ forms a partition of $\{1,\ldots,n\}$. Moreover, each $w_i$ can send at most $B$ units of flow of commodity~1 to $t_1$ and $t_1$ does not receive units of flow of commodity~1 from vertices outside the set $\{v_1, \dots, v_n\}$. Hence, $\sum_{j \in I_i} a_j \leq B$ for all $i \in [k]$. Hence, the instance of {\sc Bin Packing} is a YES-instance.
    
    Finally, note that $|V(G)| = O(n)$, that $V(G) \setminus \{v_1, \dots, v_n\}$ forms a vertex cover of size $k + 4$, and that we can construct the instance in time $n^{O(1)}$. We conclude that since {\sc Bin Packing} is NP-hard for binary weights and $k=2$~\cite{GarJohn1990}, \UILCFME{2} is NP-hard for binary weights and vertex cover size 6. We also conclude that since {\sc Bin Packing} is W[1]-hard for unary weights when parameterised by $k$ \cite{JansenKMS13}, \UILCFME{2} is W[1]-hard for unary weights, when parameterised by the vertex cover size.
\end{proof}

The reduction to the directed case can be readily seen.

\DiILCFME*
\begin{proof}
The proof immediately follows from the proof of Theorem~\ref{thm:MonoEdge}, by directing each edge from left to right (direction as in Figure~\ref{fig:BinToMono}).
\end{proof}

\section{Algorithms}
In this section, we complement our hardness results with two algorithmic results.

\subsection{Parameterisation by Weighted Tree Partition Width}
We first give an FPT-algorithm for \ILCF{\ell} parameterised by weighted tree partition width. This algorithm assumes that a tree partition of the input graph is given. There is an algorithm by Bodlaender et al.~\cite{BodlaenderGJ22} that for any graph $G$ and  integer $w$, runs in time $\mathrm{poly}(w) \cdot n^2$ and either outputs a tree partition of $G$ of width $\mathrm{poly}(w)$ or outputs that $G$ has no tree partition of width at most~$w$. By some simple tricks, this can be expanded to approximate weighted tree partition width as well, at the expense of a slightly worse polynomial in $w$. An approximately optimal tree partition of this form would be sufficient as input to our algorithm.

\UellCF*
\begin{proof}
    We will describe a dynamic-programming algorithm on a given tree partition $(T, (B_x)_{x \in V(T)})$. Let $r \in V(T)$ be some node, that we will designate as the root of the tree $T$. For convenience, we first attach a node to every leaf, with an empty bag.
    
    We will create a table $\tau$, where every entry is indexed by a node $x$ of the tree partition and a collection $\mathbf{f}_x$ of functions $f^i_x$, one function for every commodity $i \in [\ell]$.  We will refer to $\mathbf{f}_x$ as a flow profile and use the superscript $i$ to refer to the flow function for commodity $i$ in the profile. The function
    \[
    f^i_x : B_{p(x)} \to [-\tb, \tb],
    \]
    where $p(x)$ the parent node of $x$, indicates for every $v\in B_{p(x)}$ the net difference between the amount of flow of commodity $i$ that $v$ receives from (indicated by a positive value) or sends to (indicated by a negative value) the vertices in the bag $B_x$, in the current partial solution. That is, $f^i_x(v)$ models the value of $\sum_{u \in B_x} (f^i(uv) - f^i(vu))$, where $f$ denotes the current partial solution. Notice that this sum has value in $[-\tb,\tb]$, as the sum over all capacities of edges between bags $B_x$ and $B_{p(x)}$ is at most $\tb$. The content of each table entry will be a boolean that indicates whether there exists a partial flow on the graph considered up to $x$ that is consistent with the indices of the table entry.
    
    We will build the table $\tau$, starting at the leaves of the tree, for which we assumed the corresponding bags to be empty sets, and working towards the root.
    If $x$ is a leaf in the tree partition, we set $\tau[x, \{\emptyset, \dots, \emptyset\}] = \text{True}$, where we denote by $\emptyset$, the unique function with the empty set as domain. Otherwise, $x$ is some node with children $y_1, \dots, y_t$. We will group these children $y_i$ in equivalence classes $\xi$, defined by the equivalence relation $y \sim y'$ if and only if $\tau[y, \mathbf{f}_x] = \tau[y', \mathbf{f}]$ for every flow profile $\mathbf{f}$. Note that there are at most $2^{\tb^{\ell(2\tb + 1)}}$ such equivalence classes, with at most $\tb^{\ell(2\tb + 1)}$ possible flow profiles $\mathbf{f}_{y_j} = (f^1_{y_j}, \dots, f^\ell_{y_j})$ for every child $y_j$ of $x$.
    
    We will now describe an integer linear program that determines the value of $\tau[x, \mathbf{f}_x]$ for a given flow profile $\mathbf{f}_x$. We define a variable $X_{\xi, \mathbf{g}}$ as the number of sets in class $\xi$ whose in- and outflow we choose to match flow profile $\mathbf{g}$\footnote{Throughout the proof, if we sum over pairs $\xi, \mathbf{g}$, we only sum over flow profiles that are valid for bags in $\xi$. Alternatively, we can set any invalid $X_{\xi,\mathbf{g}}$ to $0$ beforehand.}. We also define a variable $Y_e^i$ for each edge inside the bag $B_x$ or between $B_x$ and its parent bag, which indicates the flow of commodity $i$ on this edge. We will denote by $N^{in}(v)$ and $N^{out}(v)$ the set of in-neighbours and out-neighbours of $v$, respectively, restricted to $B_x \cup B_{p(x)}$. We now add constraints for the following properties, for every commodity $i \in [\ell]$. Flow conservation for all vertices $v$ in the bag $B_x$, that are not a sink/source for commodity $i$:
    \[
    \sum \limits_{u \in N^{in}(v)} Y_{uv} + \sum \limits_{\xi, \mathbf{g}} X_{\xi,\mathbf{g}} \cdot g^i(v) = \sum \limits_{u \in N^{out}(v)} Y_{uv}.
    \] 
    The flow of commodity $i$ from a source $s_i$ (if $s_i \in B_x
    $):
    \[
    -\sum_{u \in N^{in}(s_i)} Y_{us_i} + \sum_{u \in N^{out}(s_i)} Y_{us_i} - \sum \limits_{\xi, \mathbf{g}} X_{\xi,\mathbf{g}} \cdot g^i(s_i) = d_i
    \] 
    The flow of commodity $i$ to a sink $t_i$ (if $t_i \in B_x$):
    \[
    \sum_{u \in N^{in}(t_i)} Y_{ut_i} - \sum_{u \in N^{out}(t_i)} Y_{ut_i} + \sum \limits_{\xi, \mathbf{g}} X_{\xi,\mathbf{g}} \cdot g^i(t_i) = d_i
    \] 
    The desired flow to a vertex $v$ in the parent bag:
    \[
    \sum_{u \in N^{in}(v) \setminus B_{p(x)}} Y_{uv} - \sum_{u \in N^{out}(v)\setminus B_{p(x)}} Y_{uv} = f_{x}^i(v).
    \]
    Edge capacities and non-negative flow:
    \[
    0 \leq \sum_{i=1}^{\ell}Y_e^i \leq c(e)
    \]
    The number of flow profiles of each type from each class matches the number of bags in that class:
    \[
    \sum \limits_{g : B_x \to [-\tb, \tb]} X_{\xi, \mathbf{g}} = |\xi|.
    \]
    $X_{\xi, \mathbf{g}}$ must be a non-negative integer:
    \[
    X_{\xi, \mathbf{g}} \in \mathbb{N}_0
    \]
    We then use an algorithm of Frank and Tardos~\cite[Theorem~5.3]{Frank1987} to solve the ILP in time $N^{2.5N + o(N)}$, where $N$ is the number of variables in the ILP. This number is dominated by the number of variables $X_{\xi, \mathbf{g}}$, of which there are $O(2^{\tb^{\ell(2\tb + 2)}})$. We thus find a running time of $2^{\tb^{\ell(2\tb + 2)}(2.5 \cdot 2^{\tb^{\ell(2\tb + 2)}} + o(2^{\tb^{\ell(2\tb + 2)}}))}$. If the ILP has a feasible solution, we set $\tau[x, \mathbf{f}_x] = \text{True}$; otherwise, we set $\tau[x, \mathbf{f}_x] = \text{False}$. We solve $\tb^{\ell(2\tb + 1)}$ such ILP's per bag in the decomposition and thus find a total running time of $O(2^{2^{\tb^{3\ell b}}})$
    
    Once we reach the root bag, we use a similar ILP to compute the flow on the root bag, finding a final solution. We find a total running time of $2^{2^{\tb^{3\ell b}}} n^{O(1)}$.
\end{proof}

Note that with some minor changes to the ILP (flow variables can be negative and there is no distinction between in/out edges), this proof also works in the undirected case. Thus, we also find the following result.

\UnILCFtb*
\begin{proof}
Adjusting the ILP so that flow variables can be negative and there is no distinction between in/out edges, we can see the proof of Theorem~\ref{thm:dwtpw} extends to the undirected case.
\end{proof}

\subsection{Parameterisation by Vertex Cover} \label{subsec:AlgVC}
In this section, we give an approximation algorithm for the \ILCF{2} problem, parameterised by the vertex cover number of a graph. Since we consider a decision problem, the use of the term `approximation' requires some explanation, since technically the output of an exact algorithm is simply `true' or `false'. In our case, we still give a Boolean output, but for the question of whether there is a flow with values in some range around the demands. We make this more precise with the following theorem:

\vcapprox*
\begin{proof}
Let $X$ be a vertex cover of the graph $G$. 
We can assume that one of size $\vc(G)$ is given as part of the input or we can compute a $2$-approximation in polynomial time by the folklore algorithm (attributed to Gavril and to Yannakakis; see e.g.~\cite{CormenLRS}). In the latter case, the proof still holds, but the (hidden) constant is worse by a factor~$8$.
Write  $Y=V\setminus X$.

We first compute, in polynomial time, a solution to the LP relaxation of the problem. This gives us a fractional flow $f$. If this flow does not meet the demands, then we can immediately answer that there exists no integral flow that meets the demands. 
We will now argue that we can transform this flow such that the total value remains the same, but the number of arcs that is assigned a non-integral value is $O(|X|^3)$.

We say that an arc $e$ is a \emph{mono-fractional arc} if there is exactly one commodity~$i \in [2]$ such that $f^i(e) \not\in \mathbb{N}_0$. We say that arc $e$ is an \emph{bi-fractional arc} if
for both commodities $i \in [2]$, we have  $f^i(e) \not\in \mathbb{N}_0$. An arc is \emph{fractional} if it is mono-fractional
or bi-fractional. An arc that is not fractional is called \emph{integral}.

We will show that a fractional flow $f$ can be transformed to
a flow with the same values for both commodities, such that there
are at most $O(|X|^3)$ fractional arcs. 
The case analysis is 
tedious, with many cases that are similar but handled slightly
differently. In each case, small changes are made to the flow, but all vertices in $X$ will have the same inflow and outflow.

Note that if $v$ is incident to an arc with fractional flow
for commodity $i$, then there must be another arc with $v$
as endpoint with fractional flow for commodity $i$, due to the flow conservation laws. We can have an incoming and an outgoing arc, two (or more) incoming arcs, or two (or more) outgoing arcs. With arcs that can be mono-fractional or bi-fractional, this gives in total twelve cases. These are
illustrated in Figure~\ref{fig:vertex cover cases}.
\begin{figure}
    \centering
    \includegraphics{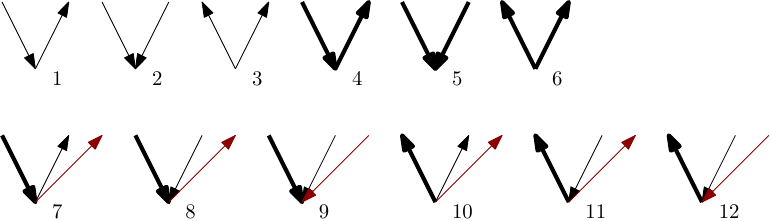}
    \caption{Different cases for vertices incident to fractional arcs. Fat edges are bi-fractional; different
    coloured non-fat edges are for different commodities.}
    \label{fig:vertex cover cases}
\end{figure}

For vertices $v\in Y$, we distinguish the following cases
\begin{enumerate}
    \item $v$ has an incoming and an outgoing mono-fractional arc for the same commodity %Case 1
    \item $v$ has two incoming mono-fractional arcs for the same commodity %Case 2
    \item $v$ has two outgoing mono-fractional arcs for the same commodity %Case 3
    \item $v$ has an incoming and an outgoing bi-fractional arc
    \item $v$ has two incoming bi-fractional arcs
    \item $v$ has two outgoing bi-fractional arcs
    \item $v$ has an incoming bi-fractional arc, and for each commodity, an outgoing mono-fractional arc.
    \item $v$ has an incoming bi-fractional arc, for one commodity an incoming mono-fractional arc, and for the other commodity an outgoing mono-fractional arc.
    \item $v$ has an incoming bi-fractional arc, and for each commodity, an incoming mono-fractional arc.
    \item $v$ has an outgoing bi-fractional arc, and for each commodity, an outgoing mono-fractional arc.
    \item $v$ has an outgoing bi-fractional arc, for one commodity an incoming mono-fractional arc, and for the other commodity an outgoing mono-fractional arc.
    \item $v$ has an outgoing bi-fractional arc, and for each commodity, an incoming mono-fractional arc.
\end{enumerate}

For each of the twelve cases, we have a rule. Each time,
we have two `similar' vertices in $Y$, and give a 
transformation of an optimal fractional flow to another optimal
fractional flow with fewer fractional values.
Some of these rules can be derived by a
symmetry argument from earlier rules, which still leaves
seven cases, each with a relatively simple proof.

For each arc $uv$ and commodity $i\in [2]$, write the fractional part of this flow as $g^i(uv)=f^i(uv)-\lfloor(f^i(uv)\rfloor$. Recall that all
arcs have integral capacities, so a mono-fractional arc
has residual capacity, i.e., it is possible to increase the
flow over the arc by a positive amount.

For each of the twelve cases, we have a rule that changes an optimal
fractional flow to another optimal fractional flow with fewer fractional values. Each time, we increase over some arcs the flow
of specified commodities and decrease over some arcs the flow of 
specified commodities by a value $\gamma$. 
Each of the changed values was fractional to start with; in each
case, we set 
$\gamma$ to the smallest value such that at least one edge hits
an integer value.

\begin{lemma}[Case 1 Rule]
Let $f$ be an optimal fractional solution to the given instance of \ILCF{2}. Suppose there are vertices $u,v\in X$, and
$y,z\in Y$, with $uy, yv, uz, zv$ mono-fractional for the same
resource $i\in [2]$. Then, in polynomial time, we can compute
optimal fractional solution $g$ with less fractional arcs than $f$.
\label{lemma:case1}
\end{lemma}

\begin{figure}
    \centering
    \includegraphics{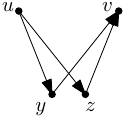}
    \caption{The subgraph for Lemma~\ref{lemma:case1}}
    \label{fig:case1rule}
\end{figure}

\begin{proof}
    Recall that all arcs have integral capacities. 
    Write $\gamma = \min\{ 1- g^i(uy), 1-g^i(yv),$ $ g^i(uz), g^i(zv) \}$.
    Note that we can increase the flow of commodity $i$ over
    arcs $uy$ and $yv$ by $\gamma$ and simultaneously decrease the
    flow of commodity $i$ over arcs $uz$ and $zv$ by $\gamma$ and
    obtain a flow with the same values. 
    By doing this, the number of fractional arcs decreases by at least one.
\end{proof}

Now, when we apply the step of Lemma~\ref{lemma:case1} exhaustively,
then we have at most $O(|X|^2)$ vertices where Case 1 applies: if
we have more than $2|X|^2$ such vertices, we 
can find a pair of vertices in $X$ and a commodity for which we
can apply the rule from Lemma~\ref{lemma:case1}.

For each of the other cases, we can use a similar argument.
For Cases 4 -- 12, we change the flow for both commodities.
For Cases 7 -- 12, we have that the resulting rules leave
$O(|X|^3)$ vertices of the specific types. 

\begin{lemma}[Case 2 Rule]
Let $f$ be an optimal fractional solution to the given instance of \ILCF{2}. Suppose there are vertices $u,v\in X$, and
$y,z\in Y$, with $uy, uz, vy, vz$ mono-fractional for the same
resource $i\in [2]$. Then, in polynomial time, we can compute
optimal fractional solution $g$ with less fractional arcs than $f$.
\label{lemma:case2}
\end{lemma}

\begin{proof}
    Set $\gamma = \min \{ 1-g^i(uy), g^i(uz), g^i(vy), 1-g^i(vz)\}$. Increase the flow (for commodity $i$) over arcs $uy$ and $vz$ by $\gamma$, and decrease the flow
    over arcs $uz$ and $vy$ by $\gamma$. All vertices have
    the same inflow and outflow, so we still have a correct flow, but with a smaller number of fractional arcs.
\end{proof}
Case 3 is similar to Case 2 but with reversed directions. The proof
is almost identical.

\begin{lemma}[Case 3 Rule]
Let $f$ be an optimal fractional solution to the given instance of \ILCF{2}. Suppose there are vertices $u,v\in X$, and
$y,z\in Y$, with $yu, zu, yv, zv$ mono-fractional for the same
resource $i\in [2]$. Then, in polynomial time, we can compute
optimal fractional solution $g$ with less fractional arcs than $f$.
\label{lemma:case3}
\end{lemma}

The next three rules deal with vertices in $Y$ incident to two
bi-fractional edges (Cases 4--6).

\begin{lemma}[Case 4 Rule]
Let $f$ be an optimal fractional solution to the given instance of \ILCF{2}. Suppose there are vertices $u,v\in X$, and
$y,z\in Y$, with $uy, yv, uz, zv$ bi-fractional. Then, in polynomial time, we can compute
optimal fractional solution $g$ with the same or fewer fractional arcs than $f$, and less bi-fractional arcs than $f$.
\label{lemma:case4}
\end{lemma}

\begin{proof}
    Set $\gamma = \min \{ 1-g^1(uy), g^2(uy), $ $
    g^1(yv), 1-g^2(yv), $ $
    g^1(uz), 1-g^2(uz),  $ $
    1-g^1(zv), g^2(zv)
    \}$.
    Change the flow as follows: 
    increase by $\gamma$ the flow of commodity 1 over arcs
    $uy$ and $yv$, and the flow of commodity 2 over arcs
    $uz$ and $zv$
    and
    decrease by $\gamma$ the flow of commodity 2 over arcs
    $uy$ and $yv$, and the flow of commodity 1 over arcs
    $uz$ and $zv$.
 One can check that each arc sends the
    same amount of total flow, and each vertex receives and sends the same amount of flow of each commodity.
    The new flow has at least one fewer fractional value.
\end{proof}

\begin{lemma}[Case 5 Rule]
Let $f$ be an optimal fractional solution to the given instance of \ILCF{2}. Suppose there are vertices $u,v\in X$, and
$y,z\in Y$, with $uy, vy, uz, vz$ bi-fractional. Then, in polynomial time, we can compute
optimal fractional solution $g$ with the same or fewer fractional arcs than $f$, and less bi-fractional arcs than $f$.
\label{lemma:case5}
\end{lemma}

\begin{proof}
    Set $\gamma = \min \{
    1-g^1(uy), g^2(uy),
    g^1(vy), 1-g^2(vy),
    g^1(uz), 1-g^2(uz),
    1-g^1(vz), g^2(vz)
    \}$.
    The flow with less fractional values is obtained by
    increasing by $\gamma$ the flow of commodity 1 over arcs $uy$ and
    $vz$ and of commodity 2 over arcs $vy$ and $uz$, and decreasing
    by $\gamma$ the flow of commodity 2 over arcs $uy$ and
    $vz$ and of commodity 1 over arcs $vy$ and $uz$.
\end{proof}

Using symmetry, we also have the following rule.

\begin{lemma}[Case 6 Rule]
Let $f$ be an optimal fractional solution to the given instance of \ILCF{2}. Suppose there are vertices $u,v\in X$, and
$y,z\in Y$, with $yu, yv, zu, zv$ bi-fractional. Then, in polynomial time, we can compute
optimal fractional solution $g$ with the same or fewer fractional arcs than $f$, and less bi-fractional arcs than $f$.
\label{lemma:case6}
\end{lemma}

\begin{lemma}
Let $f$ be an optimal fractional solution to the given instance of \ILCF{2}. Suppose none of the rules of Cases 1 -- 6 applies.
There are $O(|X|^2)$ vertices in $Y$ that are incident to at least
two mono-fractional arcs for the same commodity, or at least two
bi-fractional arcs.
\label{lemma:numberafter1-6}
\end{lemma}

\begin{proof}
If we have $9|X|^2+1$ vertices in $Y$ that are incident to at least
two mono-fractional arcs for the same commodity, or at least two
bi-fractional arcs, then we have either a commodity $i$ with
for one of Cases 1 -- 3 at least $|X|^2+1$ vertices in $Y$ of that
case for commodity $i$, or for one of the Cases 4 -- 6, at least
$|X|^2+1$ vertices in $Y$ of that case. In each of these situations,
we find a pair of vertices in $X$ with two incident vertices in $Y$ 
where we can apply a rule.
\end{proof}

It follows that the number of fractional arcs incident to a vertex
of Cases 1 -- 6 can be bounded by $O(|X|^3)$. We next give rules
that deal with pairs of vertices for Cases 7 -- 12. 

\begin{lemma}[Case 7 Rule]
    Let $f$ be an optimal fractional solution to the given instance of \ILCF{2}. Suppose there are vertices $u,v,w\in X$, and
$y,z\in Y$, with $uy, uz$ bi-fractional, $yv, zv$ mono-fractional
for commodity 1, and $yw, zw$ mono-fractional for commodity 2.
Then, in polynomial time, we can compute
optimal fractional solution $g$ with fewer fractional values.
\end{lemma}

\begin{proof}
    Let $\gamma = \min \{\
    1-g^1(uy), g^2(uy),
    1-g^1(yv), g^2(yw),
    g^1(uz), 1-g^2(uz),
    g^1(zv), 1-g^2(zw), 
    \}$.
Increase by $\gamma$ the flow of commodity $1$ over arcs $uy$ and $yv$, and of commodity 2 over arcs $uz$ and $zw$; decrease by
$\gamma$ the flow of commodity 2 over arcs $uy$ and $yw$, and of
commodity 1 over arcs $uz$ and $zv$. Again, we have an optimal flow, but at least one fractional value became integral, without creating new fractional values.
\end{proof}

\begin{figure}
    \centering
    \includegraphics{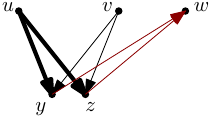}
    \caption{The subgraph for the Case 8 Rule}
    \label{fig:case8}
\end{figure}

\begin{lemma}[Case 8 Rule]
    Let $f$ be an optimal fractional solution to the given instance of \ILCF{2}. Suppose there are vertices $u,v,w\in X$, and
$y,z\in Y$, with $uy, uz$ bi-fractional, $vy, vz$ mono-fractional
for commodity $i$, and $yw, zw$ mono-fractional for commodity $3-i$.
Then, in polynomial time, we can compute
optimal fractional solution $g$ with fewer fractional values.
\end{lemma}

\begin{proof} 
Without loss of generality, assume $i=1$; otherwise switch the 
commodities.
     Let $\gamma = \min \{\
    1-g^1(uy), g^2(uy),
    g^1(vy), g^2(yw),
    g^1(uz), 1-g^2(uz),
    1-g^1(vz), 1-g^2(zw)
    \}$.
Increase by $\gamma$ the flow of commodity $1$ over arcs $uy$ and $vz$, and of commodity 2 over arcs $uz$ and $zw$; decrease by
$\gamma$ the flow of commodity 2 over arcs $uy$ and $yw$, and of
commodity 1 over arcs $vy$ and $uz$. One can again check that this flow again fulfils all conditions, but we have at least one fewer 
fractional value.
\end{proof}

\begin{lemma}[Case 9 Rule]
    Let $f$ be an optimal fractional solution to the given instance of \ILCF{2}. Suppose there are vertices $u,v,w\in X$, and
$y,z\in Y$, with $uy, uz$ bi-fractional, $yv, zv$ mono-fractional
for commodity $1$, and $yw, zw$ mono-fractional for commodity $2$.
Then, in polynomial time, we can compute
optimal fractional solution $g$ with fewer fractional values.
\end{lemma}

\begin{proof} 
    Let $\gamma = \min \{\
    1-g^1(uy), g^2(uy),
    1-g^1(yv), g^2(yw),
    g^1(uz), 1-g^2(uz),
    g^1(zv), 1-g^2(zw)
    \}$.
We increase by $\gamma$ the flow of commodity 1 over arcs
$uy$ and $yv$, and of commodity 2 over arcs $uz$ and $zw$, 
and decrease by $\gamma$ the flow of commodity 2 over arcs
$uy$ and $yw$ and of commodity 1 over arcs $uz$ and $zv$.
This gives the desired flow.
\end{proof}

For Cases 10, 11, and 12, we have similar rules. By reversing all
directions of edges, we can observe that these are symmetrical
to the rules for Cases 9, 8, and 7. We skip the details here.

\begin{lemma}
  Let $f$ be an optimal fractional solution to the given instance of \ILCF{2}.
  We can find in polynomial time 
 an optimal fractional solution $g$ with $O(|X|^3)$ fractional edges.
 \label{lemma:frac-arcs}
\end{lemma}

\begin{proof}
Apply the rules for Cases 1 -- 12 exhaustively, until none applies.

Note that each vertex in $Y$ incident to four or more fractional edges has either at least two mono-fractional edges for the same resource, or at least two bi-fractional edges. 
By Lemma~\ref{lemma:numberafter1-6}, there are $O(|X|^2)$ such vertices. 
Also, Lemma~\ref{lemma:numberafter1-6} shows we have $O(|X|^2)$ vertices
incident to two fractional edges. The number of fractional edges incident to these vertices is bounded by $O(|X|^3)$. 

It remains to bound the number of vertices that are incident to
exactly three fractional edges, and that do not fit in Cases 1--6. Such vertices necessarily belong to one
of the cases 7 -- 12. Cases 8 and 11 have two subcases, with the roles
of the two resources switched. If one of these eight cases,
there are $|X|^3+1$ vertices for which the case applies,
then there is a pair of
vertices of the same case that has the same neighbours in $X$ among their fractional arcs, and one of the rules can be applied. It follows
that there are at most $8|X|^3$ vertices in $Y$ that are incident
to exactly three fractional edges, and not belong to Cases 1--6.

Thus, the total number of fractional edges is bounded by $O(|X|^3)$.
\end{proof}

Let $g$ be the flow obtained as described above (Lemma~\ref{lemma:frac-arcs}).
We now compute an integer flow $h$ from $g$ as follows. 

For each commodity $i\in [2]$, do the following. Take the (standard, 1-commodity) flow network $G$ with source $s_i$, sink $t_i$, and for
each arc $e$, capacity $\lfloor g^i(e) \rfloor$, and compute with the
Ford-Fulkerson algorithm (or a similar flow algorithm) 
an optimal $s_i-t_i$ flow. Let the resulting flow be $h^i$.

We claim that $h^1$ and $h^2$ form together the desired integer
2-commodity flow where both commodities differ an additive term of $O(k^3)$ from the optimal flow.

The network for commodity $i$ with capacities $g^i(e)$ has a fractional
$s_i-t_i$-flow of optimal value, say $\alpha_i$, namely the flow
$g^i$. As there are $O(k^3)$ edges with $g^i(e)$ fractional,
rounding down these values decreases the total of all capacities by
$O(k^3)$, so the network for commodity $i$ with capacities $\lfloor g^i(e) \rfloor$ has an optimal value for fractional flows
$\alpha_i - O(k^3)$, but as here, all capacities are integers,
this equals the optimal value for integer flows, and a flow with
such optimal integer value is found by the Ford-Fulkerson algorithm.
The result now follows.
\end{proof}

Note that this algorithm also works for the undirected case, if we use an undirected LP and interpret the directions of the arcs in the various cases as the (net) direction of flow. Thus we also find the following result.

\Unvcapprox*
\begin{proof}
    We adjust the algorithm of Theorem~\ref{theorem:vcapprox} to use an undirected LP and interpret the directions of the arcs in the various cases as the (net) direction of flow. Then, the result holds for undirected graphs.
\end{proof}

\section{Conclusions}
\label{section:conclusions}
The parameterised complexity analysis of integer multicommodity flow 
shows that the problem is already hard for several natural parameterisations, e.g., treewidth and pathwidth, even when there are
only two commodities. The XNLP- and XALP-completeness imply that the
problems have XP algorithms but which are likely also to use $\Omega(n^{f(k)})$ space by the Slice-wise Polynomial Space Conjecture (see Conjecture~\ref{con:slice}). Moreover, the XNLP- and XALP-completeness results imply that the problems are W[t]-hard via Lemma~\ref{lem:wt}.

We end the paper with some open problems.
A number of cases for undirected graphs remain unresolved. We conjecture that for several such cases, the complexity results will be analogue to the directed case. A notable open case is \UILCF{2}, which we conjecture
is NP-complete for graphs with a pathwidth bound, but Theorem~\ref{theorem:undirected3} only gives the result with three commodities.

We also conjecture that \ILCF{2} is fixed parameter tractable with
the vertex cover number as parameter, possibly by using
a dynamic programming
algorithm that only needs to investigate solutions that are `close'
to the approximate solution found by 
Theorem~\ref{theorem:vcapprox}.

Finally, we believe that the problem may be interesting to investigate on certain graph classes, for example planar graphs of bounded treewidth or in general on graphs of treewidth or pathwidth below the bounds given by our hardness results.

\bibliographystyle{abbrvurl}
\bibliography{references}
\end{document}